\documentclass[11pt]{article}

\usepackage{fullpage}%
\usepackage{authblk}

\usepackage{amsmath}
\usepackage{amsthm}
\usepackage{amsfonts}
\usepackage{color,graphicx}
\usepackage[usenames,dvipsnames]{xcolor}
\usepackage{verbatim}
\usepackage{hyperref}
\usepackage[ruled]{algorithm}

\usepackage{enumitem}

\newtheorem{theorem}{Theorem}[section]
\newtheorem{definition}{Definition}
\newtheorem{lemma}{Lemma}[section]

\newtheorem{corollary}{Corollary}[section]
\newtheorem{observation}{Observation}[section]
\newtheorem{proposition}{Proposition}[section]
\newtheorem*{conjecture}{Conjecture}

\renewcommand{\P}{\mathtt{P}}
\newcommand{\NP}{\mathtt{NP}}

\renewcommand{\Pr}{\mathrm{Pr}}
\newcommand{\E}{\mathbf{E}}
\newcommand{\R}{\mathbb{R}}

\newcommand{\MIS}{\texttt{MIS}}

\begin{document}

\title{Bilu-Linial Stability, Certified Algorithms and the Independent Set Problem}

\author[1]{\normalsize Haris Angelidakis\thanks{Part of this work was done while the author was a student at the Toyota Technological Institute at Chicago, and was supported by the National Science Foundation under the grant CCF-1718820.}}
\author[2]{Pranjal Awasthi}
\author[3]{Avrim Blum\thanks{This work was supported in part by the National Science Foundation under grants CCF-1733556, CCF-1800317, CCF-1815011, and CCF-1535967.}}
\author[4]{Vaggos Chatziafratis\thanks{Part of this work was done while the author was a student at Stanford University.}}
\author[5]{Chen Dan}
\affil[1]{\small Eindhoven University of Technology}
\affil[2]{\small Rutgers University}
\affil[3]{\small Toyota Technological Institute at Chicago}
\affil[4]{\small Northwestern University}
\affil[5]{\small Carnegie Mellon University}

\renewcommand\Authands{ and }
\date{}

\maketitle

\normalsize
\begin{abstract}
We study the classic Maximum Independent Set problem under the notion of \textit{stability} introduced by Bilu and Linial (2010): a weighted instance of Independent Set is $\gamma$-stable if it has a unique optimal solution that remains the unique optimal solution under multiplicative perturbations of the weights by a factor of at most $\gamma\geq 1$. The goal then is to efficiently recover this ``pronounced'' optimal solution exactly. In this work, we solve stable instances of Independent Set on several classes of graphs: we improve upon previous results by solving $\widetilde{O}(\Delta/\sqrt{\log \Delta})$-stable instances on graphs of maximum degree $\Delta$, $(k - 1)$-stable instances on $k$-colorable graphs and $(1 + \varepsilon)$-stable instances on planar graphs (for any fixed $\varepsilon > 0$), using both combinatorial techniques as well as LPs and the Sherali-Adams hierarchy. 

For general graphs, we present a strong lower bound showing that there are no efficient algorithms for $O(n^{\frac{1}{2} - \varepsilon})$-stable instances of Independent Set, assuming the planted clique conjecture. To complement our negative result, we give an algorithm for $(\varepsilon n)$-stable instances, for any fixed $\varepsilon > 0$.  As a by-product of our techniques, we give algorithms as well as lower bounds for stable instances of Node Multiway Cut (a generalization of Edge Multiway Cut), by exploiting its connections to Vertex Cover. Furthermore, we prove a general structural result showing that the integrality gap of convex relaxations of several maximization problems reduces dramatically on stable instances. 

Moreover, we initiate the study of \textit{certified} algorithms for Independent Set. The notion of a $\gamma$-certified algorithm was introduced very recently by Makarychev and Makarychev (2018) and it is a class of $\gamma$-approximation algorithms that satisfy one crucial property: the solution returned is optimal for a perturbation of the original instance, where perturbations are again multiplicative up to a factor of $\gamma \geq 1$ (hence, such algorithms not only solve $\gamma$-stable instances optimally, but also have guarantees even on unstable instances). Here, we obtain $\Delta$-certified algorithms for Independent Set on graphs of maximum degree $\Delta$, and $(1+\varepsilon)$-certified algorithms on planar graphs. Finally, we analyze the algorithm of Berman and F{\"{u}}rer (1994) and prove that it is a $\left(\frac{\Delta + 1}{3} + \varepsilon\right)$-certified algorithm for Independent Set on graphs of maximum degree $\Delta$ where all weights are equal to 1.
 \end{abstract}

\section{Introduction}
The Maximum Independent Set problem (simply \MIS{} from now on) is a central problem in theoretical computer science and has been the subject of extensive research over the last few decades. As a result we now have a thorough understanding of the \textit{worst-case} behavior of the problem. In general graphs, the problem is $n^{1-\varepsilon}$-hard to approximate, assuming that $\P \neq \NP$~\cite{DBLP:conf/focs/Hastad96, DBLP:conf/stoc/Zuckerman06}, and $n/2^{(\log n)^{3/4 + \varepsilon}}$-hard to approximate, assuming that $\NP \not\subseteq \texttt{BPTIME}(2^{ (\log n)^{O(1)} })$~\cite{DBLP:conf/icalp/KhotP06}. On the positive side, the current best algorithm is due to Feige~\cite{DBLP:journals/siamdm/Feige04} achieving a $\widetilde{O}(n / \log^3 n)$-approximation (the notation $\widetilde{O}$ hides some $\mathrm{poly}(\log \log n)$ factors). In order to circumvent the strong lower bounds, many works have focused on special classes of graphs, such as bounded-degree graphs (see, e.g.,~\cite{DBLP:journals/mp/AlonK98, DBLP:journals/toc/AustrinKS11, DBLP:conf/soda/Bansal15, DBLP:conf/stoc/BansalGG15, DBLP:journals/jacm/Chan16, DBLP:journals/jgaa/Halldorsson00, DBLP:journals/njc/HalldorssonR94, DBLP:journals/siamcomp/Halperin02}), planar graphs~(\cite{DBLP:journals/jacm/Baker94}) etc. In this work, we build upon this long line of research and study \MIS{} under the beyond worst-case framework introduced by Bilu and Linial~\cite{DBLP:conf/innovations/BiluL10}.

In an attempt to capture real-life instances of combinatorial optimization problems, Bilu and Linial proposed a notion of stability, which we now instantiate in the context of \MIS{} (from now on, we will always assume weighted instances of \MIS{}).
\begin{definition}[$\gamma$-perturbation~\cite{DBLP:conf/innovations/BiluL10}]\label{def:perturbation}
Let $G = (V,E,w)$, $w: V \to \R_{>0}$, be an instance of \MIS{}. An instance $G'=(V,E,w')$ is a $\gamma$-perturbation of $G$, for some parameter $\gamma \geq 1$, if for every $u \in V$ we have $w_u \leq w'_u \leq \gamma \cdot w_u$.
\end{definition}

\begin{definition}[$\gamma$-stability~\cite{DBLP:conf/innovations/BiluL10}]\label{def:stability}
Let $G = (V,E,w)$, $w: V \to \R_{>0}$, be an instance of \MIS{}. The instance $G$ is $\gamma$-stable, for some parameter $\gamma \geq 1$, if:
\begin{enumerate}
    \item it has a unique maximum independent set $I^*$,
    \item every $\gamma$-perturbation $G'$ of $G$ has a unique maximum independent set equal to $I^*$.
\end{enumerate}
\vspace*{0.1cm}
Equivalently, $G$ is $\gamma$-stable if it has an independent set $I^*$ such that $w(I^* \setminus S) > \gamma \cdot w(S \setminus I^*)$ for every feasible independent set $S \neq I^*$ (we use the notation $w(Q):= \sum_{u \in Q} w_u$ for $Q \subseteq V$).
\end{definition}

This definition of stability is motivated by the empirical observation that in many real-life instances, the optimal solution stands out from the rest of the solution space, and thus is not sensitive to small perturbations of the parameters. This suggests that the optimal solution does not change (structurally) if the parameters of the instance are perturbed (even adversarially). Observe that the smaller the so-called \emph{stability threshold} $\gamma$ is, the less severe the restrictions imposed on the instance are; for example, $\gamma=1$ is the case where we only require the optimal solution to be unique. Thus, the main goal in this framework is to recover the optimal solution in polynomial time, for as small $\gamma \geq 1$ as possible. An ``optimal'' result would translate to $\gamma$ being $1+\varepsilon$, for small $\varepsilon>0$, since assuming uniqueness of the optimal solution is not believed to make the problems easier (see, e.g.,~\cite{DBLP:journals/tcs/ValiantV86}), and thus $\varepsilon$ is unlikely to be zero. We note that perturbations are scale-invariant, and so it suffices to consider perturbations that only scale up. Moreover, we observe that an algorithm for $\gamma$-stable instances of \MIS{} solves $\gamma$-stable instances of Minimum Vertex Cover, and vice versa.

Stability was first introduced for Max Cut~\cite{DBLP:conf/innovations/BiluL10}, but the authors note that it naturally extends to other problems, such as \MIS{}, and, moreover, they prove that the greedy algorithm for \MIS{} solves $\Delta$-stable instances on graphs of maximum degree $\Delta$. The work of Bilu and Linial has inspired numerous works on stable instances of various optimization problems; we give an overview of the literature in the next page.


Prior works on stability have also studied \emph{robust} algorithms~\cite{DBLP:conf/soda/MakarychevMV14, DBLP:conf/stoc/AngelidakisMM17}; these are algorithms that either output an optimal solution or provide a polynomial-time verifiable certificate that the instance is not $\gamma$-stable (see Section~\ref{sec:prelim} for a definition). Motivated by the notion of stability, Makarychev and Makarychev~\cite{MM18} recently introduced an intriguing class of algorithms, namely $\gamma$-\emph{certified} algorithms.

\begin{definition}[$\gamma$-certified algorithm~\cite{MM18}]\label{def:certified}
An algorithm for \MIS{} is called $\gamma$-certified, for some parameter $\gamma \geq 1$, if for every instance $G = (V,E,w)$, $w: V \to \R_{>0}$, it computes
\begin{enumerate}
    \item a feasible independent set $S \subseteq V$ of $G$,
    \item a $\gamma$-perturbation $G' = (V,E,w')$ of $G$ such that $S$ is a maximum independent set of $G'$.
\end{enumerate}
\vspace*{0.1cm}
Equivalently, Condition (2) can be replaced by the following: $\gamma \cdot w(S \setminus I) \geq w(I \setminus S)$ for every independent set $I$ of $G$.
\end{definition}

We highlight that a certified algorithm works for \emph{every} instance; if the instance is $\gamma$-stable, then the solution returned is the optimal one, while if it is not stable, the solution is within a $\gamma$-factor of optimal. Hence a $\gamma$-certified algorithm is also a $\gamma$-approximation algorithm.

\paragraph{Motivation.} Stability is especially natural for problems where the given objective function may be a proxy for a true goal of identifying a hidden correct solution. For \MIS{}, a natural such scenario is applying a machine learning algorithm in the presence of pairwise constraints. Consider, for instance, an algorithm that scans news articles on the web and aims to extract events such as ``athlete X won the Olympic gold medal in Y''. For each such statement, the algorithm gives a confidence score (e.g., it might be more confident if it saw this listed in a table rather than inferring it from a free-text sentence that the algorithm might have misunderstood). But in addition, the algorithm might also know logical constraints such as ``at most one person can win a gold medal in any given event''.  These logical constraints would then become edges in a graph, and the goal of finding the most likely combination of events would become a \MIS{} problem. Stability would be natural to assume in such a setting since the exact confidence weights are somewhat heuristic, and the goal is to recover an underlying ground truth. It is also easy to see the usefulness of a certified algorithm in this setting. Given a certified algorithm that outputs a $\gamma$-perturbation, the user of the machine learning algorithm can further test and debug the system by trying to gather evidence for events on which the perturbation puts higher weight. 

\paragraph{Related Work.} There have been many works on the worst-case complexity of \MIS{} and the current best known algorithms give $\widetilde{O}(n / \log^3 n)$-approximation~\cite{DBLP:journals/siamdm/Feige04}, and $\widetilde{O}(\Delta/\log \Delta)$-approximation~\cite{DBLP:journals/jgaa/Halldorsson00, DBLP:journals/siamcomp/Halperin02, DBLP:journals/dam/KakoOHH09}), where $\Delta$ is the maximum degree. The problem has also been studied from the lens of beyond worst-case analysis. For random graphs with a planted independent set, \MIS{} is equivalent to the classic planted clique problem. Inspired by semi-random models of~\cite{blum1995coloring}, Feige and Killian~\cite{feige2001heuristics} designed SDP-based algorithms for computing large independent sets in semi-random graphs. Finally, there has been work on \MIS{} under noise~\cite{DBLP:conf/approx/MagenM09, bansal2017lp}.

The notion of Bilu-Linial stability goes beyond random/semi-random models and proposes deterministic conditions that give rise to non worst-case, real-life instances. The study of this notion has led to insights into the complexity of many problems in optimization and machine learning. For \MIS{}, Bilu~\cite{Bilu} analyzed the greedy algorithm and showed that it recovers the optimal solution for $\Delta$-stable instances of graphs of maximum degree $\Delta$. The same result is also a corollary of a general theorem about the greedy algorithm and $p$-extendible independence systems proved by Chatziafratis et al.~\cite{DBLP:conf/esa/ChatziafratisRV17}. On the negative side, Angelidakis et al.~\cite{DBLP:conf/stoc/AngelidakisMM17} showed that there is no robust algorithm for $n^{1 - \varepsilon}$-stable instances of \MIS{} on general graphs (unbounded degree), assuming that $\P \neq \NP$. 

The work of Bilu and Linial has inspired a sequence of works about stable instances of various combinatorial optimization problems. There are now algorithms that solve $O(\sqrt{\log n} \log \log n)$-stable instances of Max Cut~\cite{DBLP:conf/stacs/BiluDLS13, DBLP:conf/soda/MakarychevMV14}, $(2 - 2/k)$-stable instances of Edge Multiway Cut, where $k$ is the number of terminals~\cite{DBLP:conf/soda/MakarychevMV14, DBLP:conf/stoc/AngelidakisMM17}, and $1.8$-stable instances of symmetric TSP~\cite{DBLP:conf/sofsem/MihalakSSW11}. There has also been extensive work on stable instances of clustering problems (usually called perturbation-resilient instances) with many positive results for problems such as $k$-median, $k$-means, and $k$-center~\cite{DBLP:journals/ipl/AwasthiBS12, DBLP:journals/siamcomp/BalcanL16, DBLP:conf/icalp/BalcanHW16, DBLP:conf/stoc/AngelidakisMM17, DBLP:conf/focs/Cohen-AddadS17, Chekuri-Approx18, pmlr-v89-deshpande19a, DBLP:conf/soda/FriggstadKS19}, and more recently on MAP inference~\cite{lang2018block,lang2018optimality}.


\paragraph{Our results.} We explore the notion of stability in the context of \MIS{} and significantly improve our understanding of its behavior on stable instances; we design algorithms for stable instances on different graph classes, and also initiate the study of certified algorithms for \MIS{}. More specifically, we obtain the following results.
\begin{itemize}
\item \textbf{Planar graphs:} We show that on planar graphs, any constant stability suffices to solve the problem exactly in polynomial time. More precisely, we provide robust and certified algorithms for $(1+\varepsilon)$-stable instances of planar \MIS{}, for any fixed $\varepsilon > 0$. To obtain these results, we utilize the Sherali-Adams hierarchy, demonstrating that hierarchies may be helpful for solving stable instances.

\item \textbf{Graphs with small chromatic number or bounded degree:} We provide robust algorithms for solving $(k-1)$-stable instances of \MIS{} on $k$-colorable graphs (where the algorithm does not have access to a $k$-coloring of the graph) and ($\Delta - 1)$-stable instances of \MIS{} on graphs of maximum degree $\Delta$. Both results are based on LPs. For bounded-degree graphs, we then turn to combinatorial techniques and design a (non-robust) algorithm for $\widetilde{O}(\Delta/\sqrt{\log \Delta})$-stable instances; this is the first algorithm that solves $o(\Delta)$-stable instances. Moreover, we show that the standard greedy algorithm is a $\Delta$-certified algorithm for \MIS{}, whereas for unweighted instances of \MIS{}, the algorithm of Berman and F\"{u}rer (1994) is a $\left(\frac{\Delta + 1}{3} + \varepsilon\right)$-certified algorithm.

\item \textbf{General graphs:} For general graphs, we show that solving $o(\sqrt{n})$-stable instances is hard assuming the hardness of finding maximum cliques in a random graph. To the best of our knowledge, this is only the second case of a lower bound for stable instances of a graph optimization problem that applies to any polynomial-time algorithm and not only to robust algorithms~\cite{DBLP:conf/soda/MakarychevMV14, DBLP:conf/stoc/AngelidakisMM17} (the first being the lower bound for Max $k$-Cut~\cite{DBLP:conf/soda/MakarychevMV14}). We complement this lower bound by giving an algorithm for $(\varepsilon n)$-stable instances of \MIS{} on graphs with $n$ vertices, for any fixed $\varepsilon > 0$.

\item \textbf{Convex relaxations and stability:} We present a structural result for the integrality gap of convex relaxations of maximization problems on stable instances: if the integrality gap of a relaxation is at most $\alpha$, then it is at most $\min \left\{\alpha, 1 + \frac{1}{\beta - 1} \right\}$ for $(\alpha \beta)$-stable instances, for any $\beta > 1$. This result demonstrates a smooth trade-off between stability and the performance of a convex relaxation, 
and also implies $(1 + \varepsilon)$-estimation algorithms\footnote{An $\alpha$-estimation algorithm returns a value that is within a factor of $\alpha$ from the optimum, but not necessarily a corresponding solution that realizes this value.} for $O(\alpha/\varepsilon)$-stable instances.

\item \textbf{Node Multiway Cut:} We give the first results on stable instances of Node Multiway Cut, a strict generalization of the well-studied (under stability) Edge Multiway Cut problem~\cite{DBLP:conf/soda/MakarychevMV14, DBLP:conf/stoc/AngelidakisMM17}. In particular, we give a robust algorithm for $(k-1)$-stable instances, where $k$ is the number of terminals, and show that all negative results on stable instances of \MIS{} directly apply to Node Multiway Cut.
\end{itemize}

\paragraph{Organization of material.} Section~\ref{sec:prelim} provides definitions and related facts. Section~\ref{sec:stable} contains the algorithms for stable instances of \MIS{} on bounded-degree, small chromatic number and planar graphs. Section~\ref{sec:general-graphs} contains our results for stable instances on general graphs. Section~\ref{sec:integrality-gaps} demonstrates how the performance of convex relaxations improves as stability increases. Section~\ref{section:mc} contains a brief description of our results on the Node Multiway Cut problem. Section~\ref{sec:certified} contains various certified algorithms for \MIS{}. We conclude with a short discussion in Section~\ref{sec:summary}. Some proofs and results have been moved to the Appendix.

\section{Preliminaries and definitions}\label{sec:prelim}

Given a $\gamma$-stable instance, our goal is to design polynomial-time algorithms that recover the unique optimal solution, for as small $\gamma \geq 1$ as possible. A special class of such algorithms that is of particular interest is the class of robust algorithms, introduced by Makarychev~et al.~\cite{DBLP:conf/soda/MakarychevMV14}.

\begin{definition}[robust algorithm~\cite{DBLP:conf/soda/MakarychevMV14}]\label{def:robust}
Let $G = (V,E,w)$, $w: V \to \R_{>0}$, be an instance of \MIS{}. An algorithm $\mathcal{A}$ is a robust algorithm for $\gamma$-stable instances if:
\begin{enumerate}
    \item it always returns the unique optimal solution of $G$, when $G$ is $\gamma$-stable,
    \item it either returns an optimal solution of $G$ or reports that $G$ is not stable, when $G$ is not $\gamma$-stable.
\end{enumerate}
\end{definition}

Note that a robust algorithm is not allowed to err, while a non-robust algorithm is allowed to return a suboptimal solution, if the instance is not $\gamma$-stable. We now present a useful lemma about stable instances of \MIS{} that is used in several of our results. From now on, we denote the neighborhood of a vertex $u$ of a graph $G = (V,E)$ as $N(u) = \{v: (u,v) \in E\}$, and the neighborhood of a set $S \subseteq V$ as $N(S) = \{v \in V \setminus S: \exists u \in S \textrm{ s.t.}~(u,v) \in E\}$.
\begin{lemma}\label{lem:delete-points-inside}
Let $G = (V,E,w)$ be a $\gamma$-stable instance of \MIS{} whose optimal independent set is $I^*$. Then, for any $v \in I^*$, the induced instance $\widetilde{G} = G[V\setminus (\{v\} \cup N(v))]$ is $\gamma$-stable, and its unique maximum independent set is $I^* \setminus \{v\}$.
\end{lemma}
\begin{proof}
It is easy to see that $I^* \setminus \{v\}$ is a maximum independent set of $\widetilde{G}$. We now prove that the instance is $\gamma$-stable. Let's assume that there exists a perturbation $w'$ of $\widetilde{G}$ such that $I' \neq (I^* \setminus \{v\})$ is a maximum independent set of $\widetilde{G}$. This means that $w'(I') \geq w'(I^* \setminus \{v\})$. We now extend $w'$ to the whole vertex set $V$ by setting $w_u' = w_u$ for every $u \in \{v\} \cup N(v)$. It is easy to verify that $w'$ is a $\gamma$-perturbation for $G$. Observe that $I' \cup \{v\}$ is a feasible independent set of $G$, and we have $w'(I' \cup \{v\})  = w'(I') + w_v' \geq w'(I^* \setminus \{v\}) + w_v' = w'(I^*)$. Thus, we get a contradiction.
\end{proof}

Regarding certified algorithms (see Definition~\ref{def:certified}), it is easy to observe the following.
\begin{observation}[\cite{MM18}]
A $\gamma$-certified algorithm for \MIS{} satisfies the following:
\begin{enumerate}
    \item returns the unique optimal solution, when run on a $\gamma$-stable instance,
    \item is a $\gamma$-certified algorithm for Vertex Cover, and vice versa,
    \item is a $\gamma$-approximation algorithm for \MIS{} (and Vertex Cover).
\end{enumerate}
\end{observation}

We stress that not all algorithms for stable instances are certified, so there is no equivalence between the two notions. Some examples (communicated to us by Yury Makarychev~\cite{Mak18}) include the algorithms for stable instances of TSP~\cite{DBLP:conf/sofsem/MihalakSSW11}, Max Cut (the GW SDP with triangle inequalities), and clustering. All these algorithms solve stable instances but are not certified. Thus, designing a certified algorithm is, potentially, a harder task than designing an algorithm for stable instances.

From now on, if an algorithm for \MIS{} only returns a feasible solution $S$, it will be assumed to be ``candidate'' $\gamma$-certified that also returns the perturbed weight function $w'$ with $w_u' = \gamma \cdot w_u$ for $u \in S$ and $w_u' = w_u$, otherwise.

\section{Stable instances of \MIS{} on special classes of graphs}\label{sec:stable}

In the next few sections, we obtain algorithms for stable instances of \MIS{} on several natural classes of graphs, by using convex relaxations and combinatorial techniques.


\subsection{Convex relaxations and robust algorithms}

The starting point for the design of robust algorithms via convex relaxations is the structural result of Makarychev et al.~\cite{DBLP:conf/soda/MakarychevMV14}, that gives sufficient conditions for the integrality of convex relaxations on stable instances. We now introduce a definition and restate their theorem in the setting of \MIS{}.

\begin{definition}[$(\alpha,\beta)$-rounding]
Let $x: V \to [0,1]$ be a feasible fractional solution of a convex relaxation of \MIS{} whose objective value for an instance $G = (V, E,w)$ is $\sum_{u \in V} w_u x_u$. A randomized rounding scheme for $x$ is an $(\alpha,\beta)$-rounding, for some parameters $\alpha, \beta \geq 1$, if it always returns a feasible independent set $S$, such that the following two properties hold for every vertex $u \in V$:
\begin{enumerate}
    \item $\Pr[u \in S] \geq \frac{1}{\alpha} \cdot x_u$,
    \item $\Pr[u \notin S] \leq \beta \cdot (1 - x_u)$.
\end{enumerate}
\end{definition}

\begin{theorem}[\cite{DBLP:conf/soda/MakarychevMV14}]\label{thm:MMV}
Let $x: V \to [0,1]$ be an optimal fractional solution of a convex relaxation of \MIS{} whose objective value for an instance $G = (V, E,w)$ is $\sum_{u \in V} w_u x_u$. Suppose that there exists an $(\alpha,\beta)$-rounding for $x$, for some $\alpha,\beta \geq 1$. Then, $x$ is integral for $(\alpha \beta)$-stable instances.
\end{theorem}
For completeness, the proof of Theorem~\ref{thm:MMV} is given in Appendix~\ref{appendix:proofs-stable}. The theorem suggests a simple robust algorithm: solve the relaxation, and if the solution is integral, report it, otherwise report that the instance is not stable (observe that the rounding scheme is used only in the analysis).

\subsection{A robust algorithm for $(k-1)$-stable instances of \MIS{} on $k$-colorable graphs}

In this section, we give a robust algorithm for $(k-1)$-stable instances of \MIS{} on $k$-colorable graphs by utilizing Theorem~\ref{thm:MMV} and the standard LP for \MIS{}. For a graph $G = (V,E,w)$, the standard LP  has an indicator variable $x_u$ for each vertex $u\in V$, and is given in Figure~\ref{fig:standard-LP}.

\begin{figure}
\begin{align*}
    \max:           &\quad \sum_{u \in V} w_u x_u\\
    \textrm{s.t.:}  &\quad x_u + x_v \leq 1,     \,\,\,\,\,\; \forall (u,v) \in E,\\
                    &\quad x_u \in [0, 1],       \quad \quad \forall u \in V.
\end{align*}
\caption{The standard LP relaxation for \MIS{}.}
\label{fig:standard-LP}
\end{figure}

The corresponding polytope is half-integral~\cite{Nemhauser1975}, and so we always have an optimal solution $x$ with $x_u \in \left\{0,\frac{1}{2},1 \right\}$ for every $u \in V$. This is useful for designing $(\alpha,\beta)$-rounding schemes, as it allows us to consider randomized combinatorial algorithms and easily present them as rounding schemes.

The crucial observation that we make is that the rounding scheme in Theorem~\ref{thm:MMV} is only used in the analysis and is not part of the algorithm, and so it can run in super-polynomial time. We also note that the final (polynomial-time) algorithm does not need to have a $k$-coloring of the graph. Let $G = (V,E,w)$ be a $k$-colorable graph, and let $x$ be an optimal half-integral solution. Let $V_i = \{u \in V: x_u = i\}$ for $i \in\{0,1/2,1\}$. We consider the rounding scheme of Hochbaum~\cite{HOCHBAUM1983243} (see Algorithm~\ref{alg:k-colorable-rounding}). We use the notation $[k] = \{1, ..., k\}$.

\begin{algorithm}[h]
\begin{enumerate}
    \item Compute a $k$-coloring $f: V_{1/2} \to [k]$ of the induced graph $G[V_{1/2}]$.
    \item Pick $j$ uniformly at random from the set $[k]$, and set $V_{1/2}^{(j)} := \{u \in V_{1/2}: f(u) = j\}$.
    \item Return $S:= V_{1/2}^{(j)} \cup V_1$.
\end{enumerate}
\caption{Hochbaum's $k$-colorable rounding scheme}
\label{alg:k-colorable-rounding}
\end{algorithm}

\begin{theorem}\label{thm:color-rounding}
Let $G = (V,E,w)$ be a $k$-colorable graph. Given an optimal half-integral solution $x$, the rounding scheme of Algorithm~\ref{alg:k-colorable-rounding} is a $\left(\frac{k}{2}, \frac{2(k-1)}{k}\right)$-rounding for $x$.
\end{theorem}
\begin{proof}
The set $S$ is feasible, as there is no edge between $V_1$ and $V_{1/2}$ and $f$ is a valid coloring. For $u \in V_0$, we have $\Pr[u \in S] = 0 = x_u$ and $\Pr[u \notin S] = 1 = 1 - x_u$. For $u \in V_1$, we have $\Pr[u \in S] = 1 = x_u$ and $\Pr[u \notin S] = 0 = 1 - x_u$. Let $u \in V_{1/2}$. We have $\Pr[u \in S] \geq \frac{1}{k} = \frac{2}{k} \cdot x_u$ and $\Pr[u \notin S] \leq 1 - \frac{1}{k} = \frac{2(k-1)}{k} \cdot (1-x_u)$. The result follows.
\end{proof}

Theorems~\ref{thm:MMV} and \ref{thm:color-rounding} now imply the following theorem, which is tight.
\begin{theorem}\label{thm:colorable-robust}
The standard LP for \MIS{} is integral for $(k-1)$-stable instances of $k$-colorable graphs.
\end{theorem}



\subsection{Algorithms for stable instances of \MIS{} on bounded-degree graphs}

Throughout this section, we assume that all graphs have maximum degree $\Delta$. The only result (prior to our work) for stable instances on such graphs was using the greedy algorithm and was given by Bilu~\cite{Bilu}.
\begin{algorithm}[h]
\begin{enumerate}
    \item Let $S := \emptyset$ and $X := V$.
    \item while $(X \neq \emptyset)$:\\
            \hspace*{20pt}Set $S := S \cup \{u\}$ and $X := X \setminus (\{u\} \cup N(u))$, where $u := \arg\max_{v \in X} \{w_v\}$.
    \item Return $S$.
\end{enumerate}
\caption{The greedy algorithm for \MIS{}}
\label{alg:greedy}
\end{algorithm}

\begin{theorem}[\cite{Bilu}]\label{thm:bilu-greedy}
The greedy algorithm (see Algorithm~\ref{alg:greedy}) solves $\Delta$-stable instances of \MIS{} on graphs of maximum degree $\Delta$.
\end{theorem}

We first note that, since the maximum degree is $\Delta$, the chromatic number is at most $\Delta + 1$, and so Theorem~\ref{thm:colorable-robust} implies a robust algorithm for $\Delta$-stable instances, giving a robust analog of Bilu's result. In fact, we can slightly improve upon that by using Brook's Theorem~\cite{brooks_1941}, which states that the chromatic number is at most $\Delta$, unless the graph is complete or an odd cycle. We can then prove following theorem.
\begin{theorem}\label{thm:delta_minus_one_robust}
There exists a robust algorithm for $(\Delta - 1)$-stable instances of \MIS{}, where $\Delta$ is the maximum degree.
\end{theorem}
Before giving the proof of the above theorem, we formally state Brook's theorem.
\begin{theorem}[Brook's theorem~\cite{brooks_1941}]
The chromatic number of a graph is at most the maximum degree $\Delta$, unless the graph is complete or an odd cycle, in which case it is $\Delta + 1$.
\end{theorem}

\begin{proof}[Proof of Theorem~\ref{thm:delta_minus_one_robust}]
\MIS{} is easy to compute on cliques and cycles. Thus, by Brook's theorem, every ``interesting" instance of maximum degree $\Delta$ is $\Delta$-colorable. We can now state the following simple algorithm. If $\Delta \leq 2$, the graph is a collection of paths and cycles, and we can find the optimal solution in polynomial time. Let's assume that $\Delta > 2$. In this case, we first separately solve all $K_{\Delta + 1}$ disjoint components, if any (we pick the heaviest vertex of each $K_{\Delta + 1}$), and then solve the standard LP on the remaining graph (whose stability is the same as the stability of the whole graph). By Brook's theorem, the remaining graph is $\Delta$-colorable. If the LP is integral, we return the solution for the whole graph, otherwise we report that the instance is not stable.
\end{proof}

We now turn to non-robust algorithms and present an algorithm that solves $o(\Delta)$-stable instances, as long as the weights are polynomially-bounded integers. The core of the algorithm is a procedure that uses an $\alpha$-approximation algorithm as a black-box in order to recover the optimal solution, when the instance is stable. Let $G = (V, E, w)$ be a graph with $n = |V|$ and $w: V \to \{1, ..., \texttt{poly}(n)\}$. Let $\mathcal{A}$ denote an $\alpha$-approximation algorithm for \MIS{}. We will give an algorithm for $\gamma$-stable instances with $\gamma = \left\lceil \sqrt{2\Delta\alpha} \right \rceil$. Note that we can assume that $\alpha\leq \Delta$ and $\gamma \leq \Delta$. These assumptions hold for the rest of this section. Algorithm~\ref{alg:o(delta)-alg} is the main algorithm, and it uses Algorithm~\ref{alg:purify} as a subroutine.

\begin{algorithm}[h]
\noindent \texttt{Bounded-Alg}$(G(V,E,w))$:
\begin{enumerate}
    \item If $w(V) \leq \gamma$, then return $V$.
    \item Run $\alpha$-approximation algorithm $\mathcal{A}$ on $G$ to get an independent set $I$.
    \item Let $S := \texttt{PURIFY}(G, I, \gamma)$. 
    \item Let $S':= \texttt{Bounded-Alg}(G[V \setminus (S \cup N(S))])$.
    \item Return $S \cup S'$.
\end{enumerate}
\caption{Algorithm for $\gamma$-stable instances, where $\gamma = \left\lceil \sqrt{2\Delta\alpha} \right \rceil$}
\label{alg:o(delta)-alg}
\end{algorithm}

\begin{algorithm}[h]
\noindent\texttt{INPUT}: Graph $G = (V, E, w)$, independent set $I \subseteq V$ and factor $\gamma \geq 1$.

\begin{enumerate}
    \item Create a bipartite unweighted graph $G_0 = (L \cup R, E_0)$, where $L$ contains $\gamma \cdot w(u)$ copies of each $u \in I$ and $R$ contains $w(v)$ copies of each $v \in V \setminus I$. The set $E_0$ is defined as follows: if $(u, v)$ is an edge in $G$ with $u \in I$ and $v \notin I$, then add edges from each copy of $u$ in $L$ to each copy of $v$ in $R$.
    \item Compute a maximum cardinality matching $M$ of $G_0$.
    \item Return the set of all vertices $u \in I$ that have at least one unmatched copy in $L$ w.r.t.~$M$.
\end{enumerate}
\caption{The \texttt{PURIFY} procedure}
\label{alg:purify}
\end{algorithm}

To prove the algorithm's correctness, we need some lemmas
\begin{lemma}\label{lemma:small-weight-is-empty}
Let $G = (V, E, w)$ be $\gamma$-stable, with $w_u \geq 1$, for every $u \in V$. If $w(V) \leq \gamma$, then $E = \emptyset$.
\end{lemma}
\begin{proof}
Suppose there exists an edge $(u,v) \in E$. Let $I^*$ be the maximum independent set. We have $w(I^*) \leq \gamma$. Wlog, let's assume that $u \notin I^*$. We define the perturbation $w'$ where $w'(u) = \gamma \cdot w(u) \geq \gamma$ and $w'(q) = w(q)$, for all $q \neq u$. We have $w'(u) \geq w'(I^*)$, and so we get a contradiction.
\end{proof}
The above lemma justifies Step 1 of Algorithm~\ref{alg:o(delta)-alg}.

\begin{lemma}\label{lemma:large-opt-bounded-degree}
Let $G = (V, E, w)$ be $\gamma$-stable, and let $I^*$ be its maximum independent set. Then $w(I^*) > \frac{\gamma}{2\Delta} \cdot w(V)$.
\end{lemma}
\begin{proof}
We look at the induced subgraph $G[V \setminus I^*]$. It has maximum degree at most $\Delta - 1$, and thus it has an independent set $I'$ of weight at least $w(V \setminus I^*)/ \Delta$. By stability, and since $I^* \cap I' = \emptyset$, we get that $w(I^*) > \gamma \cdot \frac{w(V) - w(I^*)}{\Delta}$, which implies that $w(I^*) > \frac{\gamma}{\gamma + \Delta} \cdot w(V) \geq \frac{\gamma}{2\Delta} \cdot w(V)$, where the last inequality follows from the fact that $\gamma \leq \Delta$.
\end{proof}

\begin{lemma}\label{lemma:a-approx-intersects-opt}
Let $G = (V, E, w)$ be a $\gamma$-stable instance, let $I^*$ be its maximum independent set and let $I'$ be an $\alpha$-approximate independent set. Then $I^* \cap I' \neq \emptyset$.
\end{lemma}
\begin{proof}
Suppose that $I^* \cap I' = \emptyset$. Then, by stability, we have that $w(I^*) > \gamma w(I')$. This implies that $w(I^*) > \frac{\gamma}{\alpha} \cdot w(I^*)$, which is a contradiction, since $    \frac{\gamma}{\alpha} \geq \frac{\sqrt{2\Delta\alpha}}{\alpha} = \frac{\sqrt{2\Delta}}{\sqrt{a}} \geq \frac{\sqrt{2\Delta}}{\sqrt{\Delta}} = \sqrt{2} > 1.$
\end{proof}

We now analyze the \texttt{PURIFY} procedure (Algorithm~\ref{alg:purify}).
\begin{lemma}
Let $G$ be a $\gamma$-stable instance that is given as input to the \texttt{PURIFY} procedure (see Algorithm~\ref{alg:purify}), along with an $\alpha$-approximate independent set $I$, and let $I^*$ be its maximum independent set. If $I \neq I^*$, then the set $S$ returned by the procedure always satisfies the following two properties:
\begin{enumerate}
    \item $S \neq \emptyset$, 
    \item $S \subseteq I^*$.
\end{enumerate}
\end{lemma}
\begin{proof}
We first prove Property (1). Let's assume that $S = \emptyset$. This means that all vertices in $L$ are matched. By construction, this implies that $\gamma \cdot w(I) \leq w(V \setminus I)$. Since $I$ is an $\alpha$-approximation, we have that $\gamma \cdot w(I) \geq \frac{\gamma}{\alpha} \cdot w(I^*) > \frac{\gamma \cdot \gamma}{2\Delta \alpha} w(V) \geq \frac{2\Delta\alpha}{2\Delta\alpha} w(V) = w(V)$, where the second inequality is due to Lemma~\ref{lemma:large-opt-bounded-degree}. We conclude that $w(V \setminus I) > w(V)$, which is a contradiction. Thus, $S \neq \emptyset$.

We turn to Property (2). Let $A = I \setminus I^*$ and $B = I^* \setminus I$. Let $A_0 \subseteq L$ be the copies of the vertices of set $A$ in $G_0$, and let $B_0 \subseteq R$ be the copies of the vertices of set $B$ in $G_0$. We will show that for every $Z \subseteq A_0$, we have $|N(Z) \cap B_0| \geq |Z|$. To see this, let $Z \subseteq A_0$, and let $I(Z) \subseteq A$ be the distinct vertices of $A$ whose copies (not necessarily all of them) are included in $Z$. Since the instance is $\gamma$-stable, this implies that the weight of the neighbors $F \subseteq B$ of $I(Z)$ in $I^*$ is strictly larger than $\gamma \cdot w(I(Z))$. By construction, we have that $|Z| \leq \gamma \cdot w(I(Z))$, and the number of vertices in $G_0$ corresponding to vertices of $F$ is equal to $w(F)$. Moreover, all of these $w(F)$ vertices are connected with at least one vertex in $Z$, which means that $w(F) = |N(Z) \cap B_0|$. This implies that $|N(Z) \cap B_0| > |Z|$. Thus, Hall's condition is satisfied, and so there exists a perfect matching between the vertices of $A_0$ and (a subset of the vertices of) $B_0$. 

We observe now that the neighbors of all vertices in $B_0$ are only vertices in $A_0$ and not in $L \setminus A_0$. This means that any maximum matching matches all vertices of $A_0$ (otherwise, we could increase the size of the matching by matching all vertices in $A_0$). Thus, $S \subseteq I \cap I^* \subseteq I^*$.
\end{proof}

Putting everything together, and by utilizing the $\widetilde{O}(\Delta / \log \Delta)$-approximation algorithm of Halld{\'{o}}rsson~\cite{DBLP:journals/jgaa/Halldorsson00} or Halperin~\cite{DBLP:journals/siamcomp/Halperin02} as a black-box, it is easy to  prove the following theorem.
\begin{theorem}
Algorithm~\ref{alg:o(delta)-alg} correctly solves $\left\lceil \sqrt{2\Delta\alpha} \right \rceil$-stable instances in polynomial time. In particular, there is an algorithm that solves $\widetilde{O}(\Delta/\sqrt{\log \Delta})$-stable instances.
\end{theorem}

\subsection{Robust algorithms for $(1+\varepsilon)$-stable instances of \MIS{} on planar graphs}

In this section, we design a robust algorithm for $(1+\varepsilon)$-stable instances of \MIS{} on planar graphs. Theorem~\ref{thm:colorable-robust} already implies a robust algorithm for $3$-stable instances of planar \MIS{}, but we will use the Sherali-Adams hierarchy (denoted as SA from now on) to reduce this threshold down to $1 + \varepsilon$, for any fixed $\varepsilon > 0$. In particular, we show that $O(1/\varepsilon)$ rounds of SA suffice to optimally solve $(1+\varepsilon)$-stable planar instances. We will not introduce the SA hierarchy formally, and we refer the reader to the many available surveys about it (see, e.g.,~\cite{Chlamtac2012}). The $t$-th level of SA for \MIS{} has a variable $Y_S$ for every subset $S \subseteq V$ of size at most $|S| \leq t + 1$, whose intended value is $Y_S = \prod_{u \in S} x_u$, where $x_u$ is the indicator of whether $u$ belongs to the independent set. The relaxation has size $n^{O(t)}$, and thus can be solved in time $n^{O(t)}$. For completeness, we give the relaxation in Figure~\ref{fig:SA}.

\begin{figure}
\begin{align*}
    \max: &\quad \sum_{u \in V} w_u Y_{\{u\}} \\
    \textrm{s.t.:}  &\quad \sum_{T' \subseteq T} (-1)^{|T'|} \cdot \left(Y_{S \cup T' \cup \{u\}} + Y_{S \cup T' \cup \{v\}} - Y_{S \cup T'} \right) \leq 0, && \forall (u,v) \in E, |S| + |T| \leq t,\\
                    &\quad 0 \leq \sum_{T' \subseteq T} (-1)^{|T'|} \cdot Y_{S \cup T' \cup \{u\}} \leq \sum_{T' \subseteq T} (-1)^{|T'|} \cdot Y_{S \cup T'}, && \forall u \in V, |S| + |T| \leq t,\\
                    &\quad Y_{\emptyset} = 1,\\
                    &\quad Y_{S} \in [0,1], && \forall S \subseteq V, |S| \leq t + 1.
\end{align*}
\caption{The Sherali-Adams relaxation for Independent Set.}
\label{fig:SA}
\end{figure}

Our starting point is the work of Magen and Moharrami~\cite{DBLP:conf/approx/MagenM09}, which gives a SA-based PTAS for \MIS{} on planar graphs, inspired by Baker's technique~\cite{DBLP:journals/jacm/Baker94}. In particular, \cite{DBLP:conf/approx/MagenM09} gives a rounding scheme for the $O(t)$-th round of SA that returns a $(1 + O(1/t))$-approximation. In this section, we slightly modify and analyze their rounding scheme, and prove that it satisfies the conditions of Theorem~\ref{thm:MMV}. For that, we need a theorem of Bienstock and Ozbay~\cite{DBLP:journals/disopt/BienstockO04}. For any subgraph $H$ of a graph $G = (V,E)$, let $V(H)$ denote the set of vertices contained in $H$.
\begin{theorem}[\cite{DBLP:journals/disopt/BienstockO04}]\label{thm:SA-treewidth}
Let $t \geq 1$ and $Y$ be a feasible vector for the $t$-th level SA relaxation of the standard Independent Set LP for a graph $G$. Then, for any subgraph $H$ of $G$ of treewidth at most $t$, the vector $(Y_{\{u\}})_{u \in V(H)}$ is a convex combination of independent sets of $H$.
\end{theorem}
The above theorem implies that the $t$-th level SA polytope is equal to the convex hull of all independent sets of the graph, when the graph has treewidth at most $t$.

\paragraph{The rounding scheme of Magen and Moharrami~\cite{DBLP:conf/approx/MagenM09}.}
Let $G = (V,E,w)$ be a planar graph and $\{Y_S\}_{S \subseteq V: |S| \leq t + 1}$ be an optimal $t$-th level solution of SA. We denote $Y_{\{u\}}$ as $y_u$, for any $u \in V$. We first fix a planar embedding of $G$. $V$ can then be naturally partitioned into sets $V_0, V_1, ..., V_L$, for some $L \in \{0, ..., n-1\}$, where $V_0$ is the set of vertices in the boundary of the outerface, $V_1$ is the set of vertices in the boundary of the outerface after $V_0$ is removed, and so on. Note that for any edge $(u,v) \in E$, we have $u \in V_i$ and $v \in V_j$ with $|i - j| \leq 1$. We will assume that $L \geq 4$, since, otherwise, the graph is at most $4$-outerplanar and the problem can then be solved optimally~\cite{DBLP:journals/jacm/Baker94}.

Following~\cite{DBLP:journals/jacm/Baker94}, we fix a parameter $k \in \{1, ...,L\}$, and for every $i \in \{0, ..., k-1\}$, we define $B(i) = \bigcup_{j \equiv i (\textrm{mod } k)} V_j$. We now pick an index $j \in \{0, ..., k-1\}$ uniformly at random. Let $G_0 = G[V_0 \cup V_1 ... \cup V_j]$, and for $i \geq 1$, $G_i = G[\bigcup_{q = (i-1)k + j}^{ik + j} V_q]$, where for a subset $X \subseteq V$, $G[X]$ is the induced subgraph on $X$. Observe that every edge and vertex of $G$ appears in one or two of the subgraphs $\{G_i\}$, and every vertex $u \in V \setminus B(j)$ appears in exactly one $G_i$.

Magen and Moharrami observe that for every subgraph $G_i = (V(G_i), E(G_i))$, the set of vectors $\{Y_{S}\}_{S \subseteq V(G_i): |S| \leq t + 1}$ is a feasible solution for the $t$-th level SA relaxation of the graph $G_i$. This is easy to see, as the LP associated with $G_i$ is weaker than the LP associated with $G$ (on all common variables), since $G_i$ is a subgraph of $G$, and this extends to SA as well. We need one more observation: a $k$-outerplanar graph has treewidth at most $3k-1$ (see~\cite{BODLAENDER19981}). By construction, each $G_i$ is a $(k+1)$-outerplanar graph. Thus, by setting $t = 3k+2$, Theorem~\ref{thm:SA-treewidth} implies that the vector $\{y_u\}_{u \in V(G_i)}$ can be written as a convex combination of independent sets of $G_i$. Let $p_i$ be the corresponding distribution of independent sets of $G_i$, implied by $\{y_u\}_{u \in V(G_i)}$.

We now consider the following rounding scheme. For each $G_i$, we (independently) sample an independent set $S_i$ of $G_i$ according to $p_i$. Each vertex $u \in V \setminus B(j)$ belongs to exactly one $G_i$ and is included in the final independent set $S$ if $u \in S_i$. A vertex $u \in B(j)$ might belong to two different graphs $G_i, G_{i+1}$, and is included in $S$ only if $u \in S_i \cap S_{i+1}$. The algorithm then returns $S$.

Before analyzing the algorithm, we note that standard tree-decomposition based arguments show that the rounding is constructive (i.e.~polynomial-time; this fact is not needed for the algorithm for stable instances of planar \MIS{}, but will be used when designing certified algorithms).
\begin{theorem}\label{thm:planar_rounding}
The above randomized rounding scheme always returns a feasible independent set $S$, such that for every vertex $u \in V$,
\begin{enumerate}
    \item $\Pr[u \in S] \geq \frac{k-1}{k} \cdot y_u + \frac{1}{k} \cdot y_u^2$, 
    \item $\Pr[u \notin S] \leq \left(1 + \frac{1}{k} \right) \cdot (1 - y_u)$.
\end{enumerate}
\end{theorem}
\begin{proof}
It is easy to see that $S$ is always a feasible independent set. We now compute the corresponding probabilities. Since the marginal probability of $p_i$ on a vertex $u \in G_i$ is $y_u$, for any fixed $j$, for every vertex $u \in V \setminus B(j)$, we have $\Pr[u \in S] = y_u$, and for every vertex $u \in B(j)$, we have $\Pr[u \in S] \geq y_u^2$. Since $j$ is picked uniformly at random, each vertex $u \in V$ belongs to $B(j)$ with probability exactly equal to $\frac{1}{k}$. Thus, we conclude that for every vertex $u \in V$, we have $\Pr[u \in S] \geq \frac{k-1}{k} \cdot y_u + \frac{1}{k} \cdot y_u^2$, and $\Pr[u \notin S] \leq 1 - \left(\frac{k-1}{k} \cdot y_u + \frac{1}{k} \cdot y_u^2 \right) = 1 - y_u + \frac{y_u}{k} \cdot (1 - y_u) \leq \left(1 + \frac{1}{k} \right) \cdot (1 - y_u)$.
\end{proof}
The above theorem implies that the rounding scheme is a $\left(\frac{k}{k-1}, \frac{k+1}{k} \right)$-rounding. The following theorem now is a direct consequence of Theorems~\ref{thm:MMV} and~\ref{thm:planar_rounding}. 
\begin{theorem}\label{thm:planar-stable-SA}
For every $\varepsilon > 0$, the SA relaxation of $\left(3\left\lceil \frac{2}{\varepsilon} \right\rceil + 5 \right) = O(1 / \varepsilon)$ rounds is integral for $(1 + \varepsilon)$-stable instances of \MIS{} on planar graphs.
\end{theorem}
\begin{proof}
For any given $k \geq 2$, by Theorem~\ref{thm:planar_rounding}, the rounding scheme always returns a feasible independent set $S$ of $G$ that satisfies $\Pr[u \in S] \geq \frac{k-1}{k} \cdot y_u$ and $\Pr[u \notin S] \leq \left(1 + \frac{1}{k} \right) \cdot (1 - y_u)$ for every vertex $u \in V$. By Theorem~\ref{thm:MMV}, this means that $\{y_u\}_{u \in V}$ must be integral for $(1 + \frac{2}{k-1})$-stable instances. For any fixed $\varepsilon > 0$, by setting $k = \left\lceil \frac{2}{\varepsilon} \right\rceil + 1$, we get that $3\left\lceil \frac{2}{\varepsilon} \right\rceil + 5 = O(1 / \varepsilon)$ rounds of Sherali-Adams return an integral solution for $(1+\varepsilon)$-stable instances of \MIS{} on planar graphs.
\end{proof}

\section{Stable instances of \MIS{} on general graphs}\label{sec:general-graphs}

In this section, we study stable instances of general graphs. We present a strong lower bound on any algorithm (not necessarily robust) that solves $o(\sqrt{n})$-stable instances. We complement this lower bound with an algorithm that solves $(\varepsilon n)$-stable instances in time $n^{O(1/\varepsilon)}$.

\subsection{Computational hardness of stable instances of \MIS{}}\label{sec:hardness}

We show that for general graphs it is unlikely to obtain efficient algorithms for solving $\gamma$-stable instances for small values of $\gamma$. Our hardness reduction is based on the \emph{planted clique} conjecture~\cite{feldman2013statistical, meka2015sum}, which states that finding $o(\sqrt{n})$ sized planted independent sets/cliques in the Erd\H{o}s-R\'enyi graph $G\left(n,\frac 1 2 \right)$ is computationally hard. Let $G \left(n,\frac 1 2, k \right)$ denote the distribution over graphs obtained by sampling a graph from $G \left(n,\frac 1 2 \right)$ and then picking a uniformly random subset of $k$ vertices and deleting all edges among them. The conjecture is formally stated below.

\begin{conjecture}
Let $0 < \varepsilon < \frac{1}{2}$ be a constant. Suppose that an algorithm $\mathcal{A}$ receives an input graph $G$ that is either sampled from the ensemble $G\left(n, \frac{1}{2} \right)$ or $G \left(n, \frac{1}{2}, n^{\frac 1 2 - \varepsilon} \right)$. Then, no $\mathcal{A}$ that runs in time polynomial in $n$ can decide, with probability at least $\frac 4 5$, which ensemble $G$ was sampled from.
\end{conjecture}

Our lower bound follows from the observation that planted random instances are stable up to high values of $\gamma$, and this suffices to imply our main result.
\begin{theorem}\label{thm:planted-clique-lower-bound}
Let $\varepsilon > 0$ be a constant and consider a random graph $G$ on $n$ vertices generated by first picking edges according to the Erd\H{o}s-R\'enyi model $G(n,\frac 1 2)$, followed by choosing a set $I$ of vertices of size $n^{\frac 1 2 -\varepsilon}$, uniformly at random, and deleting all edges inside $I$. Then, with probability $1-o(1)$, the resulting instance is a $\Theta(n^{\frac 1 2 -\varepsilon}/ \log n)$-stable instance of \MIS{}.
\end{theorem}
\begin{proof}
Let $G = (V,E)$ be the resulting graph (we assume that all weights are set to 1). We start by stating two well-known properties of the graph $G$ that hold with probability $1-o(1)$~(\cite{alon1998finding}).
\begin{enumerate}
    \item For each vertex $u \in V \setminus I$, we have $|N(u) \cap I| \geq \frac 1 2 \cdot n^{\frac 1 2 -\epsilon} \left(1 \pm o(1) \right)$.
    \item The size of the maximum independent set in the graph $G[V \setminus I]$ is at most  $\lceil 2(1 \pm o(1))\log n \rceil$.
\end{enumerate}

Consider any other independent set $S \neq I$. By Property $1$, we have that $|I \setminus S| \geq \frac 1 2 {n^{\frac 1 2 -\varepsilon}(1 - o(1))}$. By Property $2$, we must have that $|S \setminus I| \leq 2(1 \pm o(1))\log n$. Hence, $|S| < |I|$ and furthermore, $|I \setminus S| > \gamma \cdot |S \setminus I|$ for $\gamma = \frac{n^{\frac 1 2 -\varepsilon}}{4 \log n}$. We conclude that the instance is $\left(\frac{n^{\frac 1 2 -\varepsilon}}{4 \log n}\right)$-stable.
\end{proof}

\subsection{An algorithm for $(\epsilon n)$-stable instances}

In this section, we use the algorithm for $k$-colorable graphs and the greedy algorithm as subroutines to solve $(\varepsilon n)$-stable instances on graphs of $n$ vertices, in time $n^{O(1/\varepsilon)}$. Thus, we will assume that $\varepsilon > 0$ is a fixed constant. For a graph $G$, let $\chi(G)$ be its chromatic number. We first state a well-known result about the chromatic number, known as the Welsh-Powell algorithm for coloring; for completeness, we also give its proof.

\begin{lemma}[Welsh-Powell coloring~\cite{10.1093/comjnl/10.1.85}]\label{lemma:welsh-powell}
Let $G(V,E)$ be a graph, where $n = |V|$, and let $d_1 \geq d_2 \geq ... \geq d_n$ be the sequence of its degrees in non-increasing order. Then, $\chi(G) \leq \max_i \min \{d_i + 1, i\}$.
\end{lemma}
\begin{proof}
The lemma is based on a simple observation. We consider the following greedy algorithm for coloring. Suppose that $u_1, ..., u_n$ are the vertices of the graph, with corresponding degrees $d_1 \geq d_2 \geq ... \geq d_n$. The colors are represented with the numbers $\{1, ..., n\}$. The greedy coloring algorithm colors one vertex at a time, starting from $u_1$ and concluding with $u_n$, and for each such vertex $u_i$, it picks the ``smallest" available color. It is easy to see that for each vertex $u_i$, the color that the algorithm picks is at most ``$i$". Since the algorithm picks the smallest available color, and since the vertex $u_i$ has $d_i$ neighbors, we observe that the color picked will also be at most ``$d_i + 1$". Thus, the color of vertex $u_i$ is at most $\min\{d_i + 1, i\}$. It is easy to see now that when the algorithm terminates, it will have used at most $\max_i \min\{d_i + 1, i\}$ colors, and thus $\chi(G) \leq \max_i \min\{d_i + 1, i\}$.
\end{proof}

We now derive a useful lemma that is an implication of the Welsh-Powell coloring.
\begin{lemma}\label{lemma:colors-degrees}
Let $G=(V,E)$ be a graph, with $n = |V|$. Then, for any natural number $k \geq 1$, one of the following two properties is true:
\begin{enumerate}
    \item $\chi(G) \leq \left\lceil \frac{n}{k} \right\rceil$, or
    \item there are at least $\left\lceil \frac{n}{k} \right\rceil + 1$ vertices in $G$ whose degree is at least $\left\lceil \frac{n}{k} \right\rceil$.
\end{enumerate}
\end{lemma}
\begin{proof}
Suppose that $\chi(G) > \left\lceil \frac{n}{k} \right\rceil$. Let $u_1, ..., u_n$ be the vertices of $G$, with corresponding degrees $d_1 \geq d_2 \geq ... \geq d_n$. It is easy to see that $\max_{1 \leq i \leq \left\lceil \frac{n}{k} \right\rceil} \min \{d_i + 1, i\} \leq \left\lceil \frac{n}{k} \right\rceil$. We now observe that if $d_{\left\lceil \frac{n}{k} \right\rceil + 1} < \left\lceil \frac{n}{k} \right\rceil$, then we would have $\max_{\left\lceil \frac{n}{k} \right\rceil + 1 \leq i \leq n} \min\{d_i + 1, i\} \leq \left\lceil \frac{n}{k} \right\rceil$, and thus, by Lemma~\ref{lemma:welsh-powell}, we would get that $\chi(G) \leq \left\lceil \frac{n}{k} \right\rceil$, which is a contradiction. We conclude that we must have $d_{\left\lceil \frac{n}{k} \right\rceil + 1} \geq \left\lceil \frac{n}{k} \right\rceil$, which, since the vertices are ordered in decreasing order of their degrees, implies that there are at least $\left\lceil \frac{n}{k} \right\rceil + 1$ vertices whose degree is at least $\left\lceil \frac{n}{k} \right\rceil$.
\end{proof}

\begin{lemma}\label{lem:stable-delete-point-outside}
Let $G = (V,E, w)$ be $\gamma$-stable instance of \MIS{} whose optimal independent set is $I^*$. Then, $\widetilde{G} = G[V \setminus X]$ is $\gamma$-stable, for any set $X \subseteq V \setminus I^*$.
\end{lemma}
\begin{proof}
Fix a subset $X \subseteq V \setminus I^*$. It is easy to see that any independent set of $\widetilde{G} = G[V \setminus X]$ is an independent set of the original graph $G$. Let's assume that $\widetilde{G}$ is not $\gamma$-stable, i.e. there exists a $\gamma$-perturbation $w'$ such that $I' \neq I^*$ is a maximum independent set of $\widetilde{G}$. This means that $w'(I') \geq w'(I^*)$. By extending the perturbation $w'$ to the whole vertex set $V$ (simply by not perturbing the weights of the vertices of $X$), we get a valid $\gamma$-perturbation for the original graph $G$ such that $I'$ is at least as large as $I^*$. Thus, we get a contradiction.
\end{proof}

We will now present an algorithm for $(n/k)$-stable instances of graphs with $n$ vertices, for any natural number $k \geq 1$, that runs in time $n^{O(k)}$. Thus, by setting $k = \lceil 1/\varepsilon \rceil$, for any $\varepsilon > 0$, we can solve $(\varepsilon n)$-stable instances of \MIS{} with $n$ vertices, in total time $n^{O(1/\varepsilon)}$. Let $G = (V,E,w)$ be an $(n/k)$-stable instance of \MIS{}, where $n = |V|$. The algorithm is defined recursively (see Algorithm~\ref{alg:unbounded-degree}).

\begin{algorithm}[h]
\noindent \texttt{Unbounded-Alg$(G,k)$}:
\begin{enumerate}
    \item If $k = 1$, run greedy algorithm (Algorithm~\ref{alg:greedy}) on $G$, report solution and exit.
    \item Solve standard LP relaxation for $G$ and obtain (fractional) solution $\{x_u\}_{u \in V}$.
    \item If $\{x_u\}_{u \in V}$ is integral, report solution and exit.
    \item Let $X = \{u \in V: \mathrm{deg}(u) \geq \left\lceil \frac{n}{k} \right\rceil\}$, and for each $u \in X$, let $G_u := G[V \setminus (\{u\} \cup N(u))]$.
    \item For each $u \in X$, let $S_u :=$ \texttt{Unbounded-Alg$(G_u,k-1)$}, and set $I_u := S_u \cup \{u\}$.
    \item Let $\widetilde{G} = G[V \setminus X]$ and set $\widetilde{I} :=$ \texttt{Unbounded-Alg$(\widetilde{G},k-1)$}.
    \item Return the maximum independent set among $\{I_u\}_{u \in X}$ and $\widetilde{I}$.
\end{enumerate}
\caption{The algorithm for $(n/k)$-stable instances of \MIS{}}
\label{alg:unbounded-degree}
\end{algorithm}

\begin{theorem}\label{thm:general-graphs-alg}
There exists an algorithm that solves $(\frac{n}{k})$-stable instances of \MIS{} on graphs of $n$ vertices in time $n^{O(k)}$.
\end{theorem}
\begin{proof}
We prove the theorem by induction on $k$. Let $G = (V,E,w)$, $n = |V|$, be an $(n/k)$-stable instance whose optimal independent set is $I^*$. If $k = 1$, Theorem~\ref{thm:bilu-greedy} shows that the greedy algorithm computes the optimal solution (by setting $\Delta = n - 1$), and thus our algorithm is correct.

Let $k\geq 2$, and let's assume that the algorithm correctly solves $(N/k')$-stable instances of graphs with $N$ vertices, for any $1 \leq k' < k$. We will show that it also correctly solves $(N/k)$-stable instances. By Lemma~\ref{lemma:colors-degrees}, either the chromatic number is at most $\left\lceil \frac{n}{k} \right\rceil$, or there are at least $\left\lceil \frac{n}{k} \right\rceil + 1$ vertices whose degree is at least $\left\lceil \frac{n}{k} \right\rceil$. If the chromatic number is at most $\left\lceil \frac{n}{k} \right\rceil$, then, by Theorem~\ref{thm:colorable-robust} the standard LP is integral if $G$ is $(\left\lceil \frac{n}{k} \right\rceil - 1)$-stable. We have $\left\lceil \frac{n}{k} \right\rceil - 1 \leq \left\lfloor \frac{n}{k} \right\rfloor \leq n/k$. Thus, in this case, the LP is integral and the algorithm will terminate at step (3), returning the optimal solution.

So, let's assume that the LP is not integral, which means that $\chi(G) >  \left\lceil \frac{n}{k} \right\rceil$. Thus, the set of vertices $X = \{u \in V: \mathrm{deg}(v) \geq \left\lceil \frac{n}{k} \right\rceil\}$ has size at least $|X| \geq \left\lceil \frac{n}{k} \right\rceil + 1$. Fix a vertex $u \in X$. If $u \in I^*$, then, by Lemma~\ref{lem:delete-points-inside}, $G_u$ is $(n/k)$-stable, and moreover, $I^* = \{u\} \cup I_u^*$, where $I_u^*$ is the optimal independent set of $G_u$. $G_u$ has at most $n - \left\lceil \frac{n}{k} \right\rceil - 1 \leq \left\lfloor\frac{(k-1)}{k} \cdot n \right\rfloor = n'$ vertices, and $n/k \geq n'/(k-1)$, which implies that $G_u = (V_u, E_u, w)$ is a $\left(\frac{|V_u|}{k-1} \right)$-stable instance with $|V_u|$ vertices. Thus, by the inductive hypothesis, the algorithm computes its optimal independent set $S_u \equiv I_u^*$.

There is only one case remaining, and this is the case where $X \cap I^* = \emptyset$. In this case, by Lemma~\ref{lem:stable-delete-point-outside}, $\widetilde{G} = G[V\setminus X]$ is $(n/k)$-stable. There are at most $n - \left\lceil \frac{n}{k} \right\rceil - 1$ vertices in $\widetilde{G}$, and so, by a similar argument as above, the graph $\widetilde{G} = (\widetilde{V}, \widetilde{E}, w)$ is a $\left(\frac{|\widetilde{V}|}{k-1} \right)$-stable instance with $|\widetilde{V}|$ vertices; by the inductive hypothesis, the algorithm will compute its optimal independent set.

Since the algorithm picks the best possible independent set at step (7), it will return the optimal independent set of $G$. This concludes the induction and proves correctness. Regarding the running time, we have at most $k$ levels of recursion, and at any level, each subproblem gives rise to at most $n$ new subproblems. Thus, the total running time is bounded by $\mathrm{poly}(n) \cdot n^{k + 1} = n^{O(k)}$.
\end{proof}

\section{Stability and integrality gaps of convex relaxations}\label{sec:integrality-gaps}

In this section, we state a general theorem about the integrality gap of convex relaxations of maximization problems on stable instances. In particular, we show that, even if the conditions of Theorem~\ref{thm:MMV} are not satisfied, the integrality gap still significantly decreases as stability increases.
\begin{theorem}\label{thm:ig}
Consider a relaxation for \MIS{} that assigns a value $x_u \in [0,1]$ to each vertex $u$ of a graph $G=(V,E,w)$, and its objective function is $\sum_{u \in V} w_u x_u$. Let $\alpha$ be its integrality gap, for some $\alpha > 1$. Then, its integrality gap is at most $\min\{\alpha, 1 + \frac{1}{\beta - 1}\}$ on $(\alpha \beta)$-stable instances, for any $\beta > 1$.
\end{theorem}
\begin{proof}
Let $G = (V,E,w)$ be an $(\alpha\beta)$-stable instance, let $I^*$ denote its (unique) optimal independent set and $\mathrm{OPT} = w(I^*)$ be its cost. Let $x$ being an optimal solution to the relaxation. We introduce the notation $\mathrm{CP}(S):= \sum_{u \in S} w_u x_u$ for any $S\subseteq V$. Let's assume now that $\frac{\mathrm{CP}(V)}{\mathrm{OPT}} > 1 + \frac{1}{\beta - 1}$.

We first show that $\mathrm{CP}(I^*) < (\beta - 1) \mathrm{CP}(V \setminus I^*)$. Suppose $\mathrm{CP}(I^*) \geq (\beta - 1) \mathrm{CP}(V \setminus I^*)$. We have $CP(V) = \mathrm{CP}(I^*) + \mathrm{CP}(V \setminus I^*) \leq \left(1 + \frac{1}{\beta-1}\right) \mathrm{CP}(I^*) \leq \left(1 + \frac{1}{\beta-1}\right) \mathrm{OPT}$. This contradicts our assumption, which means that $\mathrm{CP}(I^*) < (\beta - 1) \mathrm{CP}(V \setminus I^*)$. This gives $\mathrm{CP}(V \setminus I^*) > \mathrm{CP}(V)/\beta$.

We now consider the induced graph $H = G[V \setminus I^*]$. Let $S \subseteq V \setminus I^*$ be an optimal independent set of $H$. We observe that $\{x_u\}_{u \in V \setminus I^*}$ is a feasible solution for the relaxation for $H$. Since the integrality gap is at most $\alpha$, we have $w(S) \geq \mathrm{CP}(V \setminus I^*) / \alpha$. Finally, we observe that $S$ is a feasible independent set of $G$. We conclude that $w(I^*) > (\alpha\beta) \cdot w(S)$. Combining the above inequalities, we get $\mathrm{OPT} = w(I^*) > \mathrm{CP}(V)$, and thus we get a contradiction.
\end{proof}

We stress that the above result is inherently non-constructive. Nevertheless, it suggests estimation algorithms for stable instances of \MIS{}, such as the following, which is is a direct consequence of Theorem~\ref{thm:ig} and the results of~\cite{DBLP:journals/jgaa/Halldorsson00, DBLP:journals/siamcomp/Halperin02}.
\begin{corollary}
For any fixed $\varepsilon > 0$, the Lovasz  $\theta$-function SDP has integrality gap at most $1 + \varepsilon$ on $\widetilde{O}\left(\frac{1}{\varepsilon} \cdot \frac{\Delta}{\log \Delta} \right)$-stable instances of \MIS{} of maximum degree $\Delta$.
\end{corollary}

We note that the theorem naturally extends to many other maximization graph problems, and is particularly interesting for relaxations that require super-constant stability for the recovery of the optimal solution (e.g.,~the Max Cut SDP has integrality gap $1+ \varepsilon$ for $(2/\varepsilon)$-stable instances although the integrality gap drops to exactly 1 for $\Omega(\sqrt{\log n} \cdot \log \log n)$-stable instances).

In general, such a theorem is not expected to hold for minimization problems, but, in our case, \MIS{} gives rise to its complementary Minimum Vertex Cover problem, and it turns out that we can prove a very similar result for Minimum Vertex Cover as well. For its proof (see Appendix~\ref{appendix:proofs-integrality-gaps}), we utilize Lemma~\ref{lem:delete-points-inside} and a standard trick that exploits the half-integrality property. More formally, we prove the following.
\begin{theorem}\label{thm:VC-estimation}
Suppose that there exists a convex relaxation of \MIS{} whose objective function is $\sum_{u \in V} w_u x_u$ and its integrality gap (w.r.t.~\MIS{}) is $\alpha$. Then, there exists a $\min \left\{2, 1+ \frac{1}{\beta - 2} \right\}$-estimation algorithm for $(\alpha\beta)$-stable instances of Minimum Vertex Cover, for any $\beta > 2$.
\end{theorem}

We can now easily obtain the following corollary.
\begin{corollary}[\cite{DBLP:journals/jgaa/Halldorsson00, DBLP:journals/siamcomp/Halperin02} + Theorem~\ref{thm:VC-estimation}]
For every fixed $\varepsilon > 0$, there exists a $(1 + \varepsilon)$-estimation algorithm for $\widetilde{O}\left(\frac{1}{\varepsilon} \cdot \frac{\Delta}{\log \Delta} \right)$-stable instances of minimum Vertex Cover of maximum degree $\Delta$, where the notation $\widetilde{O}$ hides some $\mathrm{poly}(\log \log \Delta)$ factors.
\end{corollary}

\section{Other applications of half-integrality for stable instances: the Node Multiway Cut problem}\label{section:mc}

In this section, we diverge a bit from \MIS{} and study another well-known case of an $\NP$-hard problem whose corresponding LP relaxation always has an optimal half-integral solution, namely the Node Multiway Cut problem. The Node Multiway Cut problem is a strict generalization of the Edge Multiway Cut problem; the latter has been extensively studied in the Bilu-Linial stability framework and has been one of the ``success stories" of the framework, where the current best algorithm solves $(2-2/k)$-stable instances (here, $k$ is the number of terminals)~\cite{DBLP:conf/soda/MakarychevMV14, DBLP:conf/stoc/AngelidakisMM17}. Thus, it is a natural question whether we can tackle stable instances of the more general Node Multiway Cut problem.

The Node Multiway Cut problem is a harder problem than the Edge Multiway Cut problem. In particular, the Edge Multiway Cut problem reduces in an approximation preserving fashion to the Node Multiway Cut problem~\cite{DBLP:journals/jal/GargVY04}. The problem is polynomially solvable for $k = 2$ and $\mathtt{APX}$-hard for $k \geq 3$. For every $k \geq 3$, the Node Multiway Cut admits a $2(1 - 1/k)$-approximation algorithm and, moreover, the standard path-based LP relaxation always has a half-integral optimal solution~\cite{DBLP:journals/jal/GargVY04}.

The Node Multiway Cut problem also turns out to be intimately related to our problem of study (or, rather, its complementary problem), Vertex Cover. More precisely, \cite{DBLP:journals/jal/GargVY04} gives an approximation-preserving reduction from minimum Vertex Cover to minimum Node Multiway Cut, which implies that, assuming $\P \neq \NP$, there is no $(\sqrt{2} - \varepsilon)$-approximation algorithm for Node Multiway Cut~\cite{DBLP:journals/eccc/KhotMS18}, and assuming UGC, there is no $(2 - \varepsilon)$-approximation algorithm~\cite{DBLP:journals/jcss/KhotR08}.

Here, we obtain the first results for stable instances of Node Multiway Cut. In particular, we prove the following two theorems.
\begin{theorem}\label{thm:mc-integrality}
The standard LP relaxation for Node Multiway Cut is integral for $(k - 1)$-stable instances, where $k$ is the number of terminals.
\end{theorem}

\begin{theorem}\label{thm:lower-bounds-mc}
\hfill
\begin{enumerate}
\itemsep0em
    \item For every constant $\varepsilon > 0$, there is no robust algorithm for $\gamma$-stable instances of minimum Node Multiway Cut, for $\gamma = n^{1 - \varepsilon}$, assuming that $P \neq NP$.
    \item Assuming the planted clique conjecture, there is no efficient algorithm for $\gamma$-stable instances of minimum Node Multiway Cut, for $\gamma = o(\sqrt{n})$.
\end{enumerate}
\end{theorem}

A complete presentation of the above results is given in Appendix~\ref{appendix:mc}.

\section{Certified algorithms for \MIS{}}\label{sec:certified}

In this section, we initiate the systematic study of certified algorithms for \MIS{}, introduced by Makarychev and Makarychev~\cite{MM18}.

\subsection{Certified algorithms using convex relaxations}
An important observation that~\cite{MM18} makes is that an approach very similar to the one used for the design of algorithms for weakly-stable instances~\cite{DBLP:conf/soda/MakarychevMV14} can be used to obtain certified algorithms. More formally, they prove the following theorem.
\begin{theorem}[\cite{MM18}]\label{thm:MM-certified}
Let $x: V \to [0,1]$ be an optimal fractional solution of a convex relaxation of \MIS{} whose objective value for an instance $G = (V, E,w)$ is $\sum_{u \in V} w_u x_u$. Suppose that there exists a polynomial-time $(\alpha,\beta)$-rounding for $x$. Then, there exists a polynomial-time $(\alpha\beta + \varepsilon)$-certified algorithm for \MIS{} on instances with integer polynomially-bounded weights (for $\varepsilon \geq 1 / \mathrm{poly}(n) > 0$).
\end{theorem}

For completeness, we present the proof in Appendix~\ref{appendix:certified}. We now combine Theorem~\ref{thm:planar_rounding} with Theorem~\ref{thm:MM-certified} and obtain the following theorem.
\begin{theorem}
There exists a polynomial-time $(1+\varepsilon)$-certified algorithm for \MIS{} on planar graphs with integer polynomially-bounded weights (for $\varepsilon \geq 1 / \mathrm{poly}(n) > 0$).
\end{theorem}



\newcommand{\I}{\mathcal{I}}
\newcommand{\sse}{\subseteq}
\newcommand{\sm}{\setminus}

\subsection{Combinatorial certified algorithms}

In this section, we study several combinatorial algorithms for \MIS{} and prove that they are certified. The first result is about the greedy algorithm.
\begin{theorem}\label{thm:certified-greedy}
The greedy algorithm (see Algorithm~\ref{alg:greedy}) is a $\Delta$-certified algorithm for \MIS{} on graphs of maximum degree $\Delta$. More generally, the greedy algorithm is a $\Delta$-certified algorithm for any instance of a $\Delta$-extendible system.
\end{theorem}
The proof of Theorem \ref{thm:certified-greedy} can be found in Appendix~\ref{appendix:proofs-certified}. The proof of the extensions to $\Delta$-extendible systems, along with the appropriate definitions, are given in Appendix~\ref{sec:extend}.
Moreover, we introduce a variant of the greedy algorithm for \MIS{} that is a  $\sqrt{\Delta^2 - \Delta + 1}$-certified algorithm; 
the improvement over the greedy is moderate for small values of $\Delta$. Thus, we present the algorithm for the special case of maximum degree $\Delta = 3$. The algorithm can then be easily generalized to any $\Delta \geq 3$. The algorithm and its analysis can be found in Appendix~\ref{appendix:improved-greedy}.

Finally, we show that the algorithm of Berman and F{\"{u}}rer~\cite{DBLP:conf/soda/BermanF94} is $\left(\frac{\Delta+1}{3} + \varepsilon \right)$-certified, when all weights are $1$. We acknowledge that the restriction to unweighted graphs limits the scope of the algorithm, but we consider this as a first step towards obtaining  $(c \Delta)$-certified algorithms, for $c < 1$.
\begin{theorem}\label{thm:berman-furer}
The Berman-F{\"{u}}rer algorithm~(\cite{DBLP:conf/soda/BermanF94}) is a $\left(\frac{\Delta + 1}{3} + \varepsilon \right)$-certified algorithm for \MIS{} on graphs of maximum degree $\Delta$, when all weights are equal to $1$.
\end{theorem}
Let $G = (V,E,w)$ be a graph of maximum degree $\Delta$, $n = |V|$, where $w_u = 1$ for every $u \in V$. We say $X$ is an improvement of $I$, if both $I$ and  $I \oplus X$ are  independent sets,  the subgraph induced by $X$ is  connected  and $I \oplus X$ is  larger  than  $I$.  (The operator  $\oplus$ denotes  the  symmetric  difference.)

The algorithm starts with a feasible independent set $I'$ and iteratively improves the solution by checking whether there exists an improvement $X$ with size $|X| \leq \sigma $. If so, it replaces $I$ by $I \oplus X$ and repeats. Otherwise, if no such improvement exists, it outputs the current independent set $I$. Assuming that $\Delta$ is a constant, the algorithm runs in polynomial time as long as $\sigma = O(\log n )$.
\begin{lemma}[\cite{DBLP:conf/soda/BermanF94}]
If $\Delta$ is a constant and $\sigma = O(\log n)$, the algorithm runs in polynomial time.
\end{lemma}

The main result can be presented as follows. Along with Definition~\ref{def:certified}, it implies Theorem~\ref{thm:berman-furer}.
\begin{lemma}\label{lemma:sym-dif-stronger}
Let $I$ be the independent set returned by the algorithm with $\sigma = 32k \Delta^{4k} \log n$ and let $S \neq I$ be any feasible independent set. Then, we have $|S \setminus I| \leq \left( \frac{\Delta + 1}{3} + \varepsilon \right) \cdot |I \setminus S|$, where $\varepsilon = \frac{1}{3k}$.
\end{lemma}
\begin{proof}
Let $\bar{S} = S \setminus I$ and $\bar{I} = I \setminus S$. First, we observe that every $u \in \bar{S}$ has at least one neighbor in $\bar{I}$, otherwise, we could improve $I$ by adding a new vertex from $\bar{S}$. We now consider the set $T = \{u \in \bar{S}: |N(u) \cap I| = 1\} \subseteq \bar{S}$. In words, $T$ is the set of elements in $\bar{S}$ that have exactly one neighbor in $I$. We also define $J = \{v \in \bar{I}: N(v) \cap T \neq \emptyset\} $ to be the set of elements of $\bar{I}$ that have at least one neighbor in $T$. We will show that $|T| \leq |J|$. 
  
To prove this, let's assume that $|T| > |J|$. Then, by the pigeonhole principle, we must have at least one vertex $v \in J$ that is connected to at least two vertices $u_1, u_2 \in T$. This implies that replacing $v$ with $u_1$ and $u_2$ would be an improvement. Thus, we get a contradiction. Now let $I_0 = \bar{I} \setminus J$ and $S_0 = \bar{S} \setminus T$. The final step of the proof is a direct consequence of Lemma 3.5 of~\cite{DBLP:conf/soda/BermanF94}, that states that if there is no improvement over $I$ of size at most $ \sigma = 32k \Delta^{4k} \log n $, then for $\varepsilon = 1/(3k)$, $|S_0| \leq \left( \frac{\Delta + 1}{3} + \varepsilon \right)|I_0|$. Recall that we have already proved $ |T| \leq |J|$. Therefore,
\begin{equation*}
    |S \setminus I| = |S_0| +|T| \leq \left( \frac{\Delta + 1}{3} + \varepsilon \right) |I_0| +|J| \leq \left( \frac{\Delta + 1}{3} + \varepsilon \right) \left( |I_0| +|J|\right) =  \left( \frac{\Delta + 1}{3} + \varepsilon \right)|I \setminus S|.
\end{equation*}
\end{proof}

\section{Summary and open problems}\label{sec:summary}

In this work we presented a finer understanding of the classic Maximum Independent Set problem on non worst-case instances. We appealed to the notion of Bilu-Linial stability and designed algorithms that efficiently find maximum independent sets in stable instances of various classes of graphs, such as planar graphs, bounded-degree graphs, small chromatic number graphs and general graphs. Furthermore, we showed that in many instances our techniques lead to certified algorithms, a natural and desirable property of any optimization algorithm. As already mentioned, a $\gamma$-certified algorithm is also a $\gamma$-approximation algorithm. Hence, an intriguing question is to investigate whether one can design certified algorithms for \MIS{} that match the best approximation guarantees. For instance, our certified algorithm for bounded-degree graphs does not match the best known approximation guarantee~\cite{DBLP:journals/siamcomp/Halperin02}. It would also be interesting to explore other sufficient conditions and properties that allow for the design of algorithms for stable instances and/or certified algorithms, as well as apply the known methods to other interesting graph classes. On the side of the lower bounds, it would be interesting to understand whether non-robust algorithms can, potentially, perform better than robust algorithms. Finally, we hope that our techniques will lead to further study of algorithms for stable instances and certified algorithms for other problems of interest.

\paragraph*{Acknowledgments.} We would like to thank Konstantin Makarychev and Yury Makarychev for kindly sharing their manuscript~\cite{MM18} with us, and Yury Makarychev and Mrinalkanti Ghosh for useful discussions.

\bibliographystyle{plain}
\bibliography{references}

\appendix

\section{Missing proofs}

\subsection{Proofs from Section~\ref{sec:prelim}}\label{appendix:proofs-prelim}

\begin{proof}[Proof of Lemma~\ref{lem:delete-points-inside}]
It is easy to see that $I^* \setminus \{v\}$ is a maximum independent set of $\widetilde{G}$. We now prove that the instance is $\gamma$-stable. Let's assume that there exists a perturbation $w'$ of $\widetilde{G}$ such that $I' \neq (I^* \setminus \{v\})$ is a maximum independent set of $\widetilde{G}$. This means that $w'(I') \geq w'(I^* \setminus \{v\})$. We now extend $w'$ to the whole vertex set $V$ by setting $w_u' = w_u$ for every $u \in \{v\} \cup N(v)$. It is easy to verify that $w'$ is a $\gamma$-perturbation for $G$. Observe that $I' \cup \{v\}$ is a feasible independent set of $G$, and we have $w'(I' \cup \{v\})  = w'(I') + w_v' \geq w'(I^* \setminus \{v\}) + w_v' = w'(I^*)$. Thus, we get a contradiction.
\end{proof}

\subsection{Proofs from Section~\ref{sec:stable}}\label{appendix:proofs-stable}

\begin{proof}[Proof of Theorem~\ref{thm:MMV}]
Let $G = (V,E,w)$ be a $\gamma$-stable instance of \MIS, where $\gamma = \alpha\beta$, whose unique optimal solution is $I^*$. Let's assume that $x$ is a non-integral optimal solution of the convex relaxation (i.e.,~there exists $u \in I^*$ such that $x_u < 1$). Then, there must exist a $u \notin I^*$ such that $x_u > 0$.

We now use the randomized rounding scheme and obtain a feasible independent set $S$. Since we have a $u \notin I^*$ such that $x_u > 0$, we get that $\Pr[u \in S] > \frac{x_u}{\alpha} > 0$, and, so, $\Pr[S \neq I^*] > 0$. By monotonicity and linearity of expectation, we get that $\E[w(I^* \setminus S)] > \gamma \E[w(S \setminus I^*)]$. Observe that
\begin{equation*}
    \E[w(I^* \setminus S)] = \sum_{u \in I^*} w_u \Pr[u \notin S] \leq \beta \sum_{u \in I^*} w_u (1 - x_u) = \beta \cdot w(I^*) - \beta \sum_{u \in I^*} w_u x_u,
\end{equation*}
and
\begin{equation*}
    \E[w(S \setminus I^*)] = \sum_{u \in V \setminus I^*} w_u \Pr[u \in S] \geq \frac{1}{\alpha} \sum_{u \in V \setminus I^*} w_u x_u.
\end{equation*}
Putting everything together, we get $w(I^*) - \sum_{u \in I^*} w_u x_u > \sum_{u \in V \setminus I^*} w_u x_u$, which implies that $w(I^*) > \sum_{u \in V} w_u x_u$. This is a contradiction, and so $x$ must indeed be integral.
\end{proof}







\subsection{Proofs from Section~\ref{sec:integrality-gaps}}\label{appendix:proofs-integrality-gaps}

\begin{proof}[Proof of Theorem~\ref{thm:VC-estimation}]
We will use a standard trick that is used for turning any good approximation algorithm for Maximum Independent Set to a good approximation algorithm for Minimum Vertex Cover. The trick is based on the fact that, if we solve the standard LP for Independent Set and look at the vertices that are half-integral, then in the induced graph on these vertices, the largest independent set is at most the size of the minimum vertex cover, and thus, any good approximate solution to Independent Set would directly translate to a good approximate solution to Vertex Cover.
	
Let $G=(V,E,w)$ be an $(\alpha\beta)$-stable instance of Vertex Cover and let $X^* \subseteq V$ be its (unique) optimal vertex cover, and $I^* = V \setminus X^*$ be its (unique) optimal independent set. We first solve the standard LP relaxation for \MIS{} and compute an optimal half-integral solution $x$. The solution $x$ naturally partitions the vertex set into three sets, $V_0 = \{u: x_u = 0\}$, $V_{1/2} = \{u: x_u = 1/2\}$ and $V_1 = \{u: x_u = 1\}$. It is well known (see~\cite{Nemhauser1975}) that $V_1 \subseteq I^*$ and $V_0 \cap I^* = \emptyset$. Thus, it is easy to see that $I^* = V_1 \cup I_{1/2}^*$, where $I_{1/2}^*$ is an optimal independent set of the induced graph $G[V_{1/2}]$ (similarly, $X^* = V_0 \cup (V_{1/2} \setminus I_{1/2}^*)$).
	
We now use the simple fact that $N(V_1) = V_0$. By iteratively applying Lemma~\ref{lem:delete-points-inside} for the vertices of $V_1$, we get that $G[V_{1/2}]$ is $(\alpha\beta)$-stable, and so it has a unique optimal independent set $I_{1/2}^*$. Let $X_{1/2}^* = V_{1/2} \setminus I_{1/2}^*$ be the unique optimal vertex cover of $G[V_{1/2}]$. It is easy to see that solution $\{x_u\}_{u \in V_{1/2}}$ (i.e. the solution that assigns value $1/2$ to every vertex) is an optimal fractional solution for $G[V_{1/2}]$. This implies that $w(I_{1/2}^*) \leq \frac{w(V_{1/2})}{2} \leq w(X_{1/2}^*)$.
	
Since $G[V_{1/2}]$ is $(\alpha\beta)$-stable, by Theorem~\ref{thm:ig} we know that the integrality gap of a convex relaxation relaxation for $G[V_{1/2}]$ is at most $\min\{\alpha, \beta / (\beta - 1)\}$. Let $A = \min\{\alpha, \beta / (\beta - 1)\}$, and let $\textrm{FRAC}$ be the optimal fractional cost of the relaxation for $G[V_{1/2}]$, w.r.t.~\MIS{}. Thus, we get that $w(I_{1/2}^*) \geq \frac{1}{A} \cdot \textrm{FRAC}$. From now on, we assume that $\beta > 2$, which implies that $1 \leq A < 2$. We now have
\begin{align*}
    w(V_{1/2}) - \textrm{FRAC} &\geq w(V_{1/2}) - A \cdot w(I_{1/2}^*) = w(V_{1/2}) - w(I_{1/2}^*) - (A - 1) \cdot w(I_{1/2}^*)\\
	                       &\geq w(X_{1/2}^*) - (A - 1) \cdot w(X_{1/2}^*) = (2 - A) \cdot w(X_{1/2}^*).
\end{align*}
We conclude that $w(X_{1/2}^*) \leq \frac{1}{2 - A} \cdot (w(V_{1/2}) - \textrm{FRAC})$. Thus, for any $\beta > 2$,
\begin{align*}
    w(V_0) + (w(V_{1/2}) - \textrm{FRAC})  &\geq w(V_0) + (2 - A) w(X_{1/2}^*) \geq (2 - A) (w(V_0) + w(X_{1/2}^*)\\
	                                   &= (2 - A) w(X^*).
\end{align*}
Since $\frac{1}{2 - A} \leq \frac{\beta - 1}{\beta - 2}$, we get that we have a $\left(1 + \frac{1}{\beta - 2} \right)$-estimation approximation algorithm for Vertex Cover on $(\alpha\beta)$-stable instances. We now combine this algorithm with any 2-approximation algorithm for Vertex Cover, and return the minimum of the two algorithms. This concludes the proof.
\end{proof}

\subsection{Proofs from Section~\ref{sec:certified}}\label{appendix:proofs-certified}

\begin{proof}[Proof of  Theorem \ref{thm:certified-greedy}]
First of all, the greedy algorithm always returns a feasible (and also maximal) solution $S$, because it starts from the empty set and greedily picks elements with maximum weight subject to being feasible. The natural perturbation $w'$ that boosts only the weights of the vertices $v \in S$ by a factor of $\Delta$, is the one we will use here to show that the Greedy is a $\Delta$-certified algorithm. More formally, $w'(v):=w(v)$ if $v\notin S$ and $w'(v):=\Delta\cdot w(v)$ if $v\in S$.
	
All we have to show is that $S$ is the optimal solution under the weight function $w'$. For the sake of contradiction, suppose $S^*$ is the optimum for the perturbed instance with $w'(S^*)> w'(S)$, where $S^*\neq S$. We order the elements in $S,S^*$ in decreasing order based on their weights $w$. We scan the elements in $S$ and let $v\in S$ be the first element that does not appear in $S^*$. Let $Z\sse S^*\sm S$ be a set of vertices, such that the set $(S^*\sm Z)\cup v$ is an independent set. By the bounded degree assumption, we know that $|Z|\le \Delta$ and by the greedy criterion we know that $w(v)\ge w(v^*)$ for any element $v^*\in Z$. Note that $w'(Z)=w(Z)$, since we didn't perturb at all the elements of $Z\sse S^*\sm S$. We conclude that $w'(v):=\Delta\cdot w(v)\ge w'(Z)=w(Z)$.
	
We can continue scanning the ordering in the same manner for all elements $v\in S\sm S^*$, ending with: $w'(S)=w'(S\cap S^*)+w'(S\sm S^*)=\Delta\cdot w(S\cap S^*)+\Delta\cdot \sum_{v\in S\sm S^*}w(v)\ge \Delta\cdot w(S\cap S^*)+\sum_{v^*\in S^*\sm S}w(v^*)=w'(S\cap S^*)+\sum_{v^*\in S^*\sm S}w'(v^*) = w'(S\cap S^*) + w'(S^*\sm S)= w'(S^*)$.
\end{proof}

\section{A greedy $\sqrt{\Delta^2 - \Delta + 1}$-certified algorithm}\label{appendix:improved-greedy}
Here, we introduce a slight variation of the greedy algorithm that gives a $\sqrt{\Delta^2 - \Delta + 1}$-certified algorithm. The improvement is moderate for small values of $\Delta$, and thus, we will present the algorithm for the special case of $\Delta = 3$; the algorithm can then be easily generalized to any degree $\Delta \geq 3$. The algorithm is based on the following lemma.
\begin{lemma}\label{lemma:sqrt-degree-3}
Let $G= (V,E,w)$ be a graph of maximum degree $\Delta = 3$. Let $\gamma = \sqrt{7}$. Let $u$ be a vertex of maximum weight (i.e. $w(u) \geq w(v)$ for every $v \in V$). Then, the following hold:
\begin{enumerate}
    \item Suppose that $|N(u)| \leq 2$. Then, there exists a $\gamma$-perturbation $G' = (V,E,w')$ with $w_u' = \gamma \cdot w_u$ and $w_v' = w_v$ for every $v \in N(u)$, and a maximum independent set $I'$ of $G'$, such that $u \in I'$.
    \item Suppose that $|N(u)| = 3$ and that $N(u)$ is not an independent set (i.e. there is at least one edge between its vertices). Then, there exists a $\gamma$-perturbation $G' = (V,E,w')$ with $w_u' = \gamma \cdot w_u$ and $w_v' = w_v$ for every $v \in N(u)$, and a maximum independent set $I'$ of $G'$, such that $u \in I'$.
    \item Suppose that $|N(u)| = 3$ and $N(u)$ is an independent set. In this case, if $\gamma \cdot w_u \geq w(N(u))$, there exists a $\gamma$-perturbation $G' = (V,E,w')$ with $w_u' = \gamma \cdot w_u$ and $w_v' = w_v$ for every $v \in N(u)$, and a maximum independent set $I'$ of $G'$ such that $u \in I'$. Otherwise, there exists a $\gamma$-perturbation $G' = (V,E,w')$ with $w_v' = \gamma \cdot w_v$ for every $v \in N(u)$, $w_q' = w_q$ for every $q \in N(N(u))$, and a maximum independent set $I'$ of $G'$ such that $N(u) \subseteq I'$.
\end{enumerate}
\end{lemma}
\begin{proof}
\hfill

\noindent 1. Let $G' = (V,E,w')$ such that $w'$ is any $\gamma$-perturbation that sets $w_u' = \gamma \cdot w_u$, and $w_v' = w_v$, for $v \in N(u)$. Then, we have that $w_u' = \gamma \cdot w_u > 2 w_u \geq w(N(u)) = w'(N(u))$. Let $I'$ be an optimal independent set of $G'$. If $u \notin I'$, then it means that $N(u) \cap I' \neq \emptyset$. It is easy to see that $(I' \setminus N(u)) \cup \{u\}$ is a feasible independent set of $G'$ whose weight is at least as large as $w'(I')$. Thus, it is an optimal independent set of $G'$.\\

\noindent 2. With a similar argument, one can prove that, even if $|N(u)| = 3$, in the case where $N(u)$ is not an independent set, the $\gamma$-perturbation $G'$, as defined in the previous case, must have an optimal independent set that contains $u$.\\

\noindent 3. Let $N(u) = \{v_1, v_2, v_3\}$, and suppose that there is no edge between the vertices of $N(u)$. We distinguish between the cases stated in the lemma:
\begin{itemize}
    \item $\gamma \cdot w_u \geq w_{v_1} + w_{v_2} + w_{v_3}$. Let $G' = (V, E,w')$ where $w'$ is any $\gamma$-perturbation that sets $w_u' = \gamma \cdot w_u$, and $w_v' = w_v$, for $v \in N(u)$. Let $I'$ be an optimal independent set of $G'$. If $u \in I'$, we are done. So, suppose that $u \notin I'$. Then, we must have $N(u) \cap I' \neq \emptyset$. We know that $w_u' \geq w(N(u)) = w'(N(u))$, and so, the set $(I' \setminus N(u)) \cup \{u\}$ is a feasible independent set whose weight is at least as large as $w'(I)$. Thus, it is an optimal independent set of $G'$.
    \item $\gamma \cdot w_u < w_{v_1} + w_{v_2} + w_{v_3}$. Let $G' = (V, E,w')$ where $w'$ is any $\gamma$-perturbation that sets $w_v' = \gamma \cdot w_v$ for every $v \in N(u)$, and $w_q' = w_q$, for $q \in N(N(u))$. Let $I'$ be an optimal independent set of $G'$ and suppose that $N(u) \not \subseteq I'$. We now consider the set $\widetilde{I} = (I' \setminus N(N(u))) \cup N(u)$. It is easy to see that $\widetilde{I}$ is a feasible independent set of $G'$. We have $w'(\widetilde{I}) \geq w'(I') - w'(N(N(u)) + w'(N(u)) \geq w'(I') - 7w_u + \gamma \cdot w(N(u)) > w'(I') - 7 w_u + \gamma^2 \cdot w_u > w'(I')$, where we used the fact that $u$ is a vertex of maximum weight in $G$, and $|N(N(u))| \leq 7$, since the maximum degree is $3$. Thus, we conclude that we must have $N(u) \subseteq I'$.
\end{itemize}
\end{proof}

The above lemma suggests an obvious greedy algorithm that runs in time $O(n\log n)$. Let $S$ be the independent set computed by the algorithm. We modify the algorithm so that it returns the independent set $S$ along with the $\gamma$-perturbation $G'=(V,E,w')$, where $w_u' = \gamma \cdot w_u$ for every $u \in S$, and $w_u' = w_u$, otherwise. It is easy to see that this is a $\sqrt{7}$-certified algorithm.
\begin{algorithm}[h]
\noindent \texttt{Modified-Greedy}($G$):
\begin{enumerate}
    \item Let $u \in V$ be a vertex of maximum weight.
    \item Let $V_0 = V \setminus (\{u\} \cup N(u))$ and $V_1 = V \setminus (N(u) \cup N(N(u))$.
    \item Using Lemma~\ref{lemma:sqrt-degree-3}:\\
            \hspace*{20pt}if $u$ is picked, then return $\{u\} \cup \texttt{Modified-Greedy}(G[V_0])$,\\
            \hspace*{20pt}else return $N(u) \cup \texttt{Modified-Greedy}(G[V_1])$.
\end{enumerate}
\caption{A modified greedy $\sqrt{7}$-certified algorithm for $\Delta = 3$.}
\label{alg:modified-greedy-certified}
\end{algorithm}

\begin{observation}
Algorithm~\ref{alg:modified-greedy-certified} generalizes to arbitrary maximum degree $\Delta \geq 3$, and it is a $\sqrt{\Delta^2 - \Delta + 1}$-certified algorithm that runs in time $\widetilde{O}(\Delta \cdot n)$.
\end{observation}

\section{$p$-extendible systems and greedy certified algorithms}\label{sec:extend}
In this section we extend some of our results to a more general family of maximization problems under $p$-extendible systems, that include \MIS{} as a special case.

\subsection{Definitions}
We start with some preliminary definitions that will be used throughout this section:
\begin{itemize}
\item $p$-systems: Suppose we are given a (finite) ground set $X$ of $m$ elements and we are also given an \textit{independence family} $\I\sse 2^X$, a family of subsets that is downward closed; that is, $A \in \I$ and $B\sse A$ imply that $B\in \I$. A set $A$ is independent iff $A\in \I$. For a set $Y\sse X$, a set $J$ is called a \textit{base} of $Y$ if $J$ is a maximal independent subset of $Y$; in other words $J\in \I$ and for each $e\in Y\sm J$, $J+e\not\in \I$. Note that $Y$ may have multiple bases and that a base of $Y$ may not be a base of a superset of $Y$. $(X,\I)$ is said to be a $p$-system if for each $Y\sse X$ the following holds:
\begin{equation*}
    \dfrac{\max_{J: J\ \text{is a base of}\ Y} |J|}{\min_{J: J\ \text{is a base of}\ Y}|J|} \le p.
\end{equation*}

There are some interesting special cases of $p$-systems (intersection of $p$ matroids, $p$-circuit-bounded and $p$-extendible families), however here our main focus will be on $p$-extendible systems.

\item An independence system $(X,\I)$ is $p$-extendible if the following holds: suppose $A\sse B, A,B \in \I$ and $A+e\in \I$, then there is a set $Z\sse B\sm A$ such that $|Z|\le p$ and $B\sm Z +e \in \I$. We note here that $p$-extendible systems make sense only for integer values of $p$, whereas $p$-systems can have $p$ being fractional.  
\item
Greedy Algorithm: Greedy starts with the empty set and greedily picks elements of $X$ that will increase its objective value by the most, while remaining feasible (according to $\I$). It is a well-known fact, that for any $p$-system, if we want to find a feasible solution $S^*\in \I$ of maximum value $f(S^*)$, then the standard greedy algorithm is a good approximation. If the weight function $f$ is additive then greedy is a $p$-approximation. If the weight function is submodular, then greedy becomes a $(p+1)$-approximation.
\end{itemize}

\subsection{Certified greedy for $p$-extendible systems}
\begin{theorem}
The greedy algorithm is a $p$-certified algorithm for any instance of a $p$-extendible system.
\end{theorem}
\begin{proof}
First of all, the Greedy always returns a feasible (and also maximal) solution $S$, because it starts from the empty set and greedily picks elements with
maximum weight subject to being feasible. The natural perturbation $w'$ that boosts only the weights of the elements $e \in S$ by a factor of $p$, is the one we will use here to show that the Greedy is a $p$-certified algorithm. More formally, $w'(e):=w(e)$ if $e\notin S$ and $w'(e):=p\cdot w(e)$ if $e\in S$.

All we have to show is that $S$ is the optimal solution under the weight function $w'$. For the sake of contradiction, suppose $S^*$ is the optimum for the perturbed instance with $w'(S^*)> w'(S)$, where $S^*\neq S$. We order the elements in $S,S^*$ in decreasing order based on their weights $w'$. We scan the elements in $S$ and let $e\in S$ be the first element that does not appear in $S^*$. Let $Z\sse S^*\sm S$ be a set of elements such that $(S^*\sm Z)\cup e \in \I$. By the $p$-extendibility property we know that $|Z|\le p$ and by the greedy criterion we know that $w(e)>w(e^*)$ for any element $e^*\in Z$. Note that $w'(Z)=w(Z)$, since we didn't perturb at all the elements of $Z\sse S^*\sm S$. We conclude that $w'(e):=p\cdot w(e)\ge w'(Z)=w(Z)$.

We can continue scanning the ordering in the same manner for all elements $e\in S\sm S^*$, ending with: $w'(S)=w'(S\cap S^*)+w'(S\sm S^*)=p\cdot w(S\cap S^*)+p\cdot \sum_{e\in S\sm S^*}w(e)\ge p\cdot w(S\cap S^*)+\sum_{e^*\in S^*\sm S}w(e^*)=w'(S\cap S^*)+\sum_{e^*\in S^*\sm S}w'(e^*) = w'(S\cap S^*) + w'(S^*\sm S)= w'(S^*)$.
\end{proof}

The above theorem is tight for the greedy algorithm, as the following proposition suggests:
\begin{proposition}
There exist $p$-extendible systems where greedy cannot be $(p-\epsilon)$-certified.
\end{proposition}
\begin{proof}
A special case of a 2-extendible system is the problem of maximum weighted matching. Consider a path of length 3 with weights $(1,1+\epsilon',1)$.
The Greedy fails to recover a certified solution if we have picked $\epsilon'$ small enough ($\epsilon'<\tfrac{\epsilon}{2-\epsilon}$). 
The proposition follows since a similar example for any value of $p$ (e.g. $p$-dimensional matching problem) and with arbitrarily large size can be constructed by repeating it.
\end{proof}

The following proposition highlights the importance of the $p$-extendibility property exploited by the greedy algorithm, by proving that greedy cannot generally be a certified algorithm for the immediate generalization of $p$-extendible systems, which are called $p$-systems:
\begin{proposition}
For $p$-systems greedy fails to be $M$-certified (for arbitrary $M>1$).
\end{proposition}
\begin{proof}
The counterexample is the same as in~\cite{DBLP:conf/esa/ChatziafratisRV17} and it is based on a knapsack constraint. 
\end{proof}

\section{The framework of Makarychev and Makarychev~\cite{MM18} for certified algorithms}\label{appendix:certified}

In this section, we describe the framework of Makarychev and Makarychev~\cite{MM18} for designing certified algorithm by using convex relaxations, which is inspired by the framework of Makarychev et al.~\cite{DBLP:conf/soda/MakarychevMV14} for solving weakly-stable instances. Since certified algorithms also ``solve" weakly-stable instances, we provide here the definition of weak stability.

\begin{definition}[weak stability~\cite{DBLP:conf/soda/MakarychevMV14}]\label{def:weakly-stable}
Let $G = (V,E,w)$ be an instance of \MIS{} with a unique optimal solution $I^*$. Let $\mathcal{N}$ be a set of feasible independent sets of $G$ such that $I^* \in \mathcal{N}$, and let $\gamma \geq 1$. The instance is $(\gamma, \mathcal{N})$-weakly stable if for every $\gamma$-perturbation $G' = (V,E,w')$, we have $w'(I^*) > w'(S)$, for every independent set $S \notin \mathcal{N}$. Equivalently, the instance is $(\gamma, \mathcal{N})$-weakly stable if $w(I^* \setminus S) > \gamma \cdot w(S \setminus I^*)$ for every independent set $S \notin \mathcal{N}$.
\end{definition}
In the above definition, the set $\mathcal{N}$ can be thought of as a neighborhood of feasible solutions of $I^*$, and the definition in that case implies that the optimal solution might change, but not too much. The algorithmic task then is to find a solution $S \in \mathcal{N}$; note that we are not given the set $\mathcal{N}$. Observe that a $\gamma$-stable instance of \MIS{} whose optimal solution is $I^*$ is $(\gamma, \{I^*\})$-weakly stable.

We now state a simple observation.
\begin{observation}
A $\gamma$-certified algorithm returns a solution $S \in \mathcal{N}$, when run on a $(\gamma, \mathcal{N})$-weakly stable instance.
\end{observation}

We are now ready to present the framework of Makarychev and Makarychev~\cite{MM18}. Let $G=(V,E,w)$ be an instance of $\MIS{}$ and $w: V \to \{1, ..., W\}$, for some integer $W = \mathrm{poly}(n)$, where $n = |V|$. In this setting, we will prove Theorem~\ref{thm:MM-certified}, but before that, we prove the following lemma.

\begin{lemma}[\cite{DBLP:conf/soda/MakarychevMV14, DBLP:conf/stoc/AngelidakisMM17, MM18}]\label{lemma:weakly-stable}
Let $x: V \to [0,1]$ be an optimal fractional solution of a convex relaxation of \MIS{} whose objective value for an instance $G = (V, E,w)$ is $\sum_{u \in V} w_u x_u$. Suppose that there exists a polynomial-time $(\alpha, \beta)$ rounding for $x$ that returns a feasible independent $S$. Then, there is an algorithm that, for any $\varepsilon > 0$, given an instance of \MIS{} and a feasible independent set $S$, does the following with probability at least $\frac{1}{2}$:
\begin{itemize}
    \item if there exists an independent set $I$ such that $w(I \setminus S) > \gamma \cdot w(S \setminus I)$, then it finds an independent set $S'$ such that
        \begin{equation*}
            w(I) - w(S') \leq \left(1 - \frac{\varepsilon}{2\alpha(\alpha\beta + \varepsilon)} \right) \left(w(I) - w(S) \right),
        \end{equation*}
    \item if $w(I \setminus S) \leq \gamma \cdot w(S \setminus I)$ for every independent set $I$, it either returns an independent set $S'$ with $w(S') > w(S)$, or certifies that $S$ is a $\gamma$-certified solution.
\end{itemize}
The algorithm's running time is $\mathrm{poly}\left(n,\alpha,\beta, \frac{1}{\varepsilon} \right)$.
\end{lemma}
\begin{proof}
We define the perturbation $G'=(V,E,w')$, where $w_u' = (\alpha\beta) \cdot w_u$, if $u \in S$, and $w_u' = w_u$, otherwise. We solve the convex relaxation for $G'$ and obtain the fractional solution $x$ to which we can apply the rounding scheme. If $\sum_{u \in V} w_u' x_u = w'(S)$, then the algorithm terminates and certifies that $S$ is a $\gamma$-certified solution, since in this case, $S$ is optimal for $G'$ (which is a $(\gamma-\varepsilon)$-perturbation of $G$). So, let's assume that $\sum_{u \in V} w_u' x_u > w'(S)$. We then apply the rounding scheme on $x$ and obtain an independent set $S'$. It is easy to see that there must exist at least one $u \notin S$ with $x_u > 0$, and so $\Pr[S'\neq S] > 0$. We have
\begin{equation*}
\begin{split}
    \E[w(S') - w(S)] &= \E[w(S' \setminus S) - w(S \setminus S')] = \E[w(S' \setminus S)] - \E[w(S \setminus S')] \\
                     &= \sum_{u\in V \setminus S} w_u \Pr[u \in S'] - \sum_{u \in S} w_u \Pr[ u \notin S'] \geq\frac{1}{\alpha} \sum_{u\in V \setminus S} w_u x_u - \beta \sum_{u \in S} w_u (1 - x_u)\\
                     &= \frac{1}{\alpha} \sum_{u\in V \setminus S} w_u' x_u - \frac{1}{\alpha} \sum_{u \in S} w_u' (1 - x_u) = \frac{1}{\alpha} \left(\sum_{u \in V} w_u' x_u - w'(S)\right).
\end{split}
\end{equation*}
Suppose now that there exists an independent set $I \neq S$ such that $w(I \setminus S) > \gamma \cdot w(S \setminus I)$. In this case, we get
\begin{equation*}
\begin{split}
    \E[w(S') - w(S)] &= \frac{1}{\alpha} \left(\sum_{u \in V} w_u' x_u - w'(S)\right) \geq \frac{1}{\alpha} \left(w'(I) - w'(S)\right) = \frac{1}{\alpha} \left(w'(I\setminus S) - w'(S\setminus I)\right)\\
	             &= \frac{1}{\alpha} \left(w(I\setminus S) - (\alpha \beta) w(S\setminus I)\right) > \frac{1}{\alpha} \left(w(I\setminus S) - \frac{\alpha \beta}{\alpha\beta + \varepsilon} \cdot w(I\setminus S)\right)\\
		     &= \frac{\varepsilon}{\alpha\beta + \varepsilon} \cdot w(I \setminus S) \geq \frac{\varepsilon}{\alpha\beta + \varepsilon} \cdot \left(w(I) - w(S) \right).
\end{split}
\end{equation*}
We conclude that $\E[w(I) - w(S')] < \left(1 - \delta \right) \cdot \left(w(I) - w(S) \right)$, where $\delta = \frac{\varepsilon}{\alpha(\alpha\beta + \varepsilon)}$. Then, by applying Markov's inequality, we get that
\begin{equation*}
    \Pr \left[w(I) - w(S') > \left(1-\frac{\delta}{2} \right)\left(w(I) - w(S) \right) \right] < \frac{1- \delta}{1 - \delta/2} = 1 - \frac{\delta}{2 - \delta} \leq 1 - \frac{\delta}{2}.
\end{equation*}
Thus, with probability at least $\delta/2$, we get an independent set $S'$ that satisfies
\begin{equation}\label{ineq:weak--stab-progress}
    w(I) - w(S') \leq \left(1-\frac{\delta}{2} \right) \left(w(I) - w(S) \right).
\end{equation}
We now repeat the rounding process $M = \frac{2\ln 2}{\delta}$ times, independently, and obtain independent sets $S_1', S_2', ..., S_M'$. Let $S'$ be the largest independent set among $S_1', S_2', ..., S_M'$. The probability that $S'$ violates inequality~(\ref{ineq:weak--stab-progress}) is at most $\left(1 - \frac{\delta}{2} \right)^M \leq e^{-\frac{\delta \cdot M}{2}} = \frac{1}{2}$. If $w(S') > w(S)$, the algorithm returns $S'$, otherwise the algorithm certifies that $S$ is a $\gamma$-certified solution.
\end{proof}


\begin{proof}[Proof of Theorem~\ref{thm:MM-certified}]
The algorithm starts with any feasible independent set $S^{(0)}$. We iteratively apply the algorithm presented in Lemma~\ref{lemma:weakly-stable} as follows: we apply the algorithm $t$ times, for some $t \geq 1$ to be specified later, in order to boost the probability of success, and pick the largest of the independent sets returned, if any. Let $S^{(1)}$ be the largest such independent set. We repeat this process $T \geq 1$ times, and obtain a sequence of independent sets $S^{(1)}, S^{(2)}, ..., S^{(T)}$, where $T$ will be specified later (we clarify that in order to get $S^{(i)}$, we will again run the algorithm $t$ times in order to boost the probability of success). We note that, in order to obtain $S^{(i+1)}$, the algorithm of Lemma~\ref{lemma:weakly-stable} is given $S^{(i)}$ as input. Thus, we apply the algorithm of Lemma~\ref{lemma:weakly-stable} at most $t \cdot T$ times, for a total running time of $\mathrm{poly}\left(n,\alpha,\beta, \frac{1}{\varepsilon}, t, T \right)$

If the algorithm at any iteration reports that some $S^{(i)}$ is $\gamma$-certified, then we return $S^{(i)}$, and the algorithm terminates. So, let's assume that the algorithm always does an improving step and finds the next set $S^{(i+1)}$. Since all weights are integers, we have $w(S^{(i+1)}) \geq w(S^{(i)}) + 1$. Thus, since $\sum_{u \in V} w_u \leq n \cdot W = \mathrm{poly}(n)$, it is clear that after polynomially many steps the algorithm must terminate by certifying that a solution $S^{(i)}$ is $\gamma$-certified.

We set $T = n \cdot W$. The only remaining thing is to decide on the value of the parameter $t$. Each iteration $i$ fails with probability at most $2^{-t}$. Thus, the probability of failure over the $T$ iterations is at most $T \cdot 2^{-t}$. Thus, by setting $t = \log (n \cdot T)$, we conclude that the algorithm fails with probability at most $1/n$. As already observed, the total running time is $\mathrm{poly}\left(n,\alpha,\beta, \frac{1}{\varepsilon}, t, T \right)$, which is polynomial in the size of the input when $\varepsilon \geq 1 \ \mathrm{poly}(n)$.

\end{proof}

\section{Stable instances of the Minimum Node Multiway Cut problem}\label{appendix:mc}

We first define the problem.
\begin{definition}[Node Multiway Cut]
Let $G = (V, E)$ be a connected undirected graph and let $T = \{s_1, ..., s_k\} \subseteq V$ be a set of terminals such that for every $i \neq j$, $(s_i, s_j) \notin E$. In the Node Multiway Cut problem, we are given a function $w: V \to \R_{>0}$ and the goal is to remove the minimum weight set of vertices $V' \subseteq V \setminus T$ such that in the induced graph $G' = G[V \setminus V']$, there is no path between any of the terminals.
\end{definition}

A $\gamma$-stable instance $G = (V,E,w)$ with terminal set $T$ is defined as expected; it has a unique optimal solution $X^* \subseteq V \setminus T$, and every $\gamma$-perturbation $G'= (V, E, w')$ of the instance has the same unique optimal solution $X^*$. We observe that it is straightforward to reprove the Theorem~\ref{thm:MMV} in the setting of Node Multiway Cut, and in particular, one can easily prove that it suffices to obtain an $(\alpha, \beta)$-rounding for a half-integral optimal solution, since such a solution always exists. We now give one such rounding for the standard LP relaxation for Node Multiway Cut given in Figure~\ref{fig:MC-LP}, that satisfies $\alpha \beta = k - 1$, where $k$ is the number of terminals.

Let $G = (V, E, w)$, $T  = \{s_1, ..., s_k\} \subseteq V$, be an instance of Node Multiway Cut. The standard LP relaxation is given in Figure~\ref{fig:MC-LP}. The LP has one indicator variable for each vertex $u \in V$. For each pair of terminals $s_i$ and $s_j$, $i < j$, let $\mathcal{P}_{ij}$ denote the set of all paths between $s_i$ and $s_j$. Let $\mathcal{P} = \bigcup_{i < j} \mathcal{P}_{ij}$.

\begin{figure}[ht]
\begin{align*}
    \min:          & \quad \sum_{u \in V \setminus T} w_u x_u \\
    \textrm{s.t.:} & \quad \sum_{u \in P} x_u \geq 1, \quad\quad\;           \forall P \in \mathcal{P},\\
                   & \quad x_{s_{i}} = 0,             \quad\quad\quad\quad \forall i \in [k],\\
                   & \quad x_u \geq 0,                \quad\quad\quad\quad\, \forall u \in V.
\end{align*}
\caption{The standard LP relaxation for Node Multiway Cut.}
\label{fig:MC-LP}
\end{figure}

We now present a rounding scheme for the LP (Algorithm~\ref{alg:half-integral-rounding-mc}) that only works for half-integral solutions. Let $\{x_u\}_{u \in V}$ be a half-integral optimal solution for the LP of Figure~\ref{fig:MC-LP}. For $i \in \{0, \frac{1}{2}, 1\}$, let $V_i = \{u \in V: x_u = i\}$. Since $x$ is half-integral, we have $V = V_0 \cup V_{1/2} \cup V_1$. For a path $P$, let $len(P) = \sum_{u \in P} x_u$. Let $\mathcal{P}_{uv}$ denote the set of all paths between two vertices $u$ and $v$. We define $d(u,v) = \min_{P \in \mathcal{P}_{uv}} len(P)$; we note that this function is not an actual metric, since we always have some $u \in V$ with $d(u,u) > 0$. We consider the following rounding scheme (see Algorithm~\ref{alg:half-integral-rounding-mc}).

\begin{algorithm}[h]
\begin{enumerate}
    \item Let $G' = G[V_0 \cup V_{1/2}]$  (if graph $G'$ has more than one connected component, we apply the rounding scheme on each connected component, separately).
    \item For each $i \in [k]$, let $B_i = \{u \in V_0: d(s_i, u) = 0\}$ and $\delta(B_i) = \{u \in V_{1/2}: \exists v \in B_i \textrm{ such that }(u,v) \in E\}$ (we note that the function $d$ is computed separately in each connected component of $G'$).
    \item Pick uniformly random $j^* \in [k]$.
    \item Return $X: = V_1 \cup (\bigcup_{i \neq j^*} \delta(B_i))$.
\end{enumerate}
\caption{An $(\alpha, \beta)$-rounding for half-integral solutions of Node Multiway Cut}
\label{alg:half-integral-rounding-mc}
\end{algorithm}

\begin{theorem}
Algorithm~\ref{alg:half-integral-rounding-mc} is an $(\alpha, \beta)$-rounding for half-integral optimal solutions of Node Multiway Cut, for some $\alpha$ and $\beta$, with $\alpha \beta = k - 1$. More precisely, given an optimal half-integral solution $\{x_u\}_{u \in V}$, it always returns a feasible solution $X \subseteq V \setminus T$ such that for each vertex $u \in V \setminus T$, the following two conditions are satisfied:
\begin{enumerate}
    \item $\Pr[u \in X] \leq \alpha \cdot x_u$,
    \item $\Pr[u \notin X] \geq \frac{1}{\beta} \cdot (1 - x_u)$,
\end{enumerate}
with $\alpha = \frac{2(k - 1)}{k}$ and $\beta = \frac{k}{2}$.
\end{theorem}
\begin{proof}
We first show that $X$ is always a feasible solution. It is easy to see that $s_i \notin X$ for every $i \in [k]$. Let's fix now a path $P$ between $s_i$ and $s_j$. If there exists a vertex $u \in P$ such that $x_u = 1$, then clearly the algorithm ``cuts" this path, since $X$ contains all vertices whose LP value is 1. So, let's assume that for every $u \in P$ we have $x_u \in \{0,1/2\}$. Observe that the whole path $P$ is contained in the graph $G'$. Since $x_{s_t} = 0$ for every $t \in [k]$, we have $s_i \in B_i$ and $s_j \in B_j$ and we know that at least one of the sets $\delta(B_i)$ or $\delta(B_j)$ will be included in the solution. The LP constraints imply that $\sum_{q \in P} x_q \geq 1$. Thus, there are at least 2 vertices in $P$ whose LP value is exactly $1/2$. So, we start moving along the path $P$ from $s_i$ to $s_j$, and let $q_1 \in P$ be the first vertex with $x_{q_1} = 1 / 2$. Similarly, we start moving along the path from $s_j$ to $s_i$, and let $q_2 \in P$ be the first vertex with $x_{q_2} = 1 / 2$. Our assumption implies that $q_1 \neq q_2$. Clearly, $d(s_i, q_1) = d(s_j, q_2) = 1/2$, and it is easy to see that $q_1 \in \delta(B_i)$ and $q_2 \in \delta(B_j)$. Thus, at least one of the vertices $q_1$ or $q_2$ will be included in the final solution $X$. We conclude that the algorithm always returns a feasible solution.  

We will now show that the desired properties of the rounding scheme are satisfied with $\alpha\beta = k - 1$. For that, we first prove that $\bigcup_{i \in [k]} \delta(B_i) = V_{1/2}$, and moreover, each $u \in V_{1/2}$ belongs to exactly one set $\delta(B_i)$. By definition $\bigcup_{i \in [k]} \delta(B_i) \subseteq V_{1/2}$. Let $u \in V_{1/2}$. It is easy to see that there must exist at least one path $P$ between two terminals such that $u \in P$ and $x_v < 1$ for every $v \in P$, since otherwise we could simply set $x_u = 0$ and still get a feasible solution with lower cost. Let's assume now that $u \notin \bigcup_{i \in [k]} \delta(B_i)$. This means that for any path $P$ between two terminals $s_i$ and $s_j$ such that $u \in P$ and $x_v < 1$ for every $v \in P$, if we start moving from $s_i$ to $s_j$, we will encounter at least one vertex $q_1 \neq u$ with $x_{q_1} = 1/2$, and similarly, if we start moving from $s_j$ to $s_i$, we will encounter at least one vertex $q_2 \neq u$ with $x_{q_2} = 1/2$. Since this holds for any two terminals $s_i$ and $s_j$, it is easy to see that we can set $x_u = 0$ and get a feasible solution with a smaller cost. Thus, we get a contradiction. This shows that $\bigcup_{i \in [k]} \delta(B_i) = V_{1/2}$. We will now prove that for every $u \in V_{1/2}$ there exists a unique $i \in [k]$ such that $u \in \delta(B_i)$. Suppose that $u \in \delta(B_i) \cap \delta(B_j)$, for some $i \neq j$. Let $q_1 \in B_i$ such that $(u, q_1) \in E$, and let $q_2 \in B_j$ such that $(u, q_j) \in E$. Let $P_1$ be a shortest path between $s_i$ and $q_1$, and let $P_2$ be a shortest path between $s_j$ and $q_2$. We now consider the path $P' = P_1 \cup \{u\} \cup P_2$. This is indeed a valid path in $G'$ between $s_i$ and $s_j$. It is easy to see that $\sum_{v \in P'} x_v = 1/2$, and so an LP constraint is violated. Again, we get a contradiction, and thus, we conclude that for each $u \in V_{1/2}$ there exists exactly one $i \in [k]$ such that  $u \in \delta(B_i)$.

We are almost done. We will now verify that the two conditions of the rounding scheme are satisfied. Let $u \in V \setminus T$. If $x_u = 1$, then $u$ is always picked and we have $\Pr[u \in X] = 1 = x_u$ and $\Pr[u \notin X] = 0 = 1 - x_u$. If $x_u = 0$, then the vertex $u$ will never be picked, and so $\Pr[u \in X] = 0 = x_u$ and $\Pr[u \notin X] = 1 = 1 - x_u$. So, let's assume now that $x_u = 1/2$. By the previous discussion, $u \in \delta(B_i)$ for some unique $i \in [k]$. Since each set $\delta(B_i)$ is not included in the solution with probability $1/k$, we get that $\Pr[u \notin X] = \frac{1}{k} = \frac{2}{k} \cdot (1 - x_u)$, and $\Pr[u \in X] = \frac{k - 1}{k} = \frac{2(k-1)}{k} \cdot x_u$. Thus, the rounding scheme satisfies the desired properties with $\alpha\beta = \frac{2(k-1)}{k} \cdot \frac{k}{2} = k-1$.
\end{proof}

The above theorem, combined with the adaptation of Theorem~\ref{thm:MMV} for the problem directly gives Theorem~\ref{thm:mc-integrality}. Mimicking the techniques of~\cite{DBLP:conf/soda/MakarychevMV14}, we can also prove the following theorem about weakly-stable instances.
\begin{theorem}
There is a polynomial-time algorithm that, given a $(k - 1 + \delta, \mathcal{N})$-weakly-stable instance of Minimum Node Multiway Cut with $n$ vertices, $k$ terminals and integer polynomially-bounded weights, finds a solution $X' \in \mathcal{N}$ (for every $\delta \geq 1/\mathrm{poly}(n) > 0$).
\end{theorem}

We now prove that the above analysis is tight, i.e.~there are $(k - 1 - \varepsilon)$-stable instances for which the LP is not integral.
\begin{theorem}
For every $\varepsilon > 0$, there exist $(k - 1-\varepsilon)$-stable instances of the Node Multiway Cut problem with $k$ terminals for which the LP of Figure~\ref{fig:MC-LP} is not integral.
\end{theorem}
\begin{proof}
We consider a variation of the star graph, as shown in Figure~\ref{fig:mc-bad-example}. The graph $G = (V, E, w)$ is defined as follows:
\begin{enumerate}
    \item $V = \{s_1, ..., s_k\} \cup \{u_1, ..., u_k\} \cup \{c\}$, with $T = \{s_1, ..., s_k\}$ being the set of terminals. Observe that $|V| = 2k + 1$.
    \item $E = \{(c,u_i): i \in [k]\} \cup \{(s_i, u_i): i\in [k]\}$.
    \item For each $i \in \{1, ..., k - 1\}$, we have $w_{u_i} = 1$. We also have $w_{u_k} = k - 1 - \frac{\varepsilon}{2}$ and $w_c = k^3$.
\end{enumerate}

\begin{figure}[h]
\begin{center}
\scalebox{0.7}{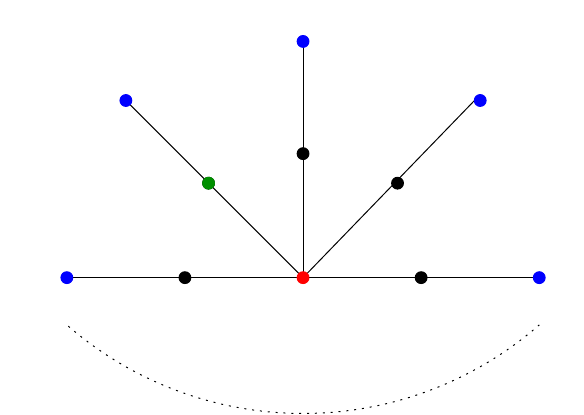}
\end{center}
\caption{An integrality gap example of a stable instance of Node Multiway Cut.}
\label{fig:mc-bad-example}
\end{figure}

It is easy to see that there is unique optimal integral solution $X^* = \{u_i : 1 \leq i \leq k - 1\}$ of cost $OPT = k - 1$. It is also clear that any feasible solution must either remove vertex $c$ or must remove at least $k - 1$ vertices from the set $\{u_1, ..., u_k\}$. A minimal solution that contains $c$ is $X_c = \{c\}$. We have $(k - 1 - \varepsilon) w(X^* \setminus X_c) < (k - 1)^2$ and $w(X_c \setminus X^*) = k^3$. Let's consider now a solution that does not contain $c$. By the previous observations, we only have to consider the solutions $Y_i = \{u_1, ..., u_k\} \setminus \{u_i\}$, $1 \leq i \leq k - 1$, and $Y_0 = \{u_1, ..., u_k\}$. For any $Y_i$, $1 \leq i \leq k - 1$, we have $(k - 1 - \varepsilon) \cdot w(X^* \setminus X_i) = (k - 1 - \varepsilon) \cdot w_{u_i} = k - 1 - \varepsilon$ and $w(Y_i \setminus X^*) = w_{u_k} = k - 1 - \varepsilon / 2$. For $Y_0$ we have $(k - 1 - \varepsilon) \cdot w(X^* \setminus Y_0) = 0$ and $w(Y_i \setminus X^*) = w(u_k) = k - 1 - \varepsilon/2$. Thus, in all cases, the stability condition is satisfied with $\gamma = k - 1 - \varepsilon$.

We now look at the LP. Let $x_{u_i} = 1/2$ for every $i \in [k]$ and let $x_c = 0$. We also set $x_{s_i} = 0$ for every $i \in [k]$. Observe that this is a feasible solution. The objective function is equal to
\begin{equation*}
    \frac{k-1}{2} + \frac{k -1 - (\varepsilon/2)}{2} = k - 1 - (\varepsilon/4) < k - 1 = OPT.
\end{equation*}
Thus, the integrality gap is strictly greater than 1, and thus, the LP is not integral.
\end{proof}

Finally, we show that if there exists an algorithm for $\gamma$-stable instances of Node Multiway Cut, then there exists an algorithm for $\gamma$-stable instances of Vertex Cover. This reduction, combined with the negative results for Vertex Cover, implies strong lower bounds on the existence of efficient algorithms for stable instances of Node Multiway Cut.
\begin{theorem}
Let $\mathcal{A}$ be an algorithm for $\gamma$-stable instances of Minimum Node Multiway Cut. Then, there exists an algorithm $\mathcal{B}$ for $\gamma$-stable instances of Minimum Vertex Cover. Moreover, if $\mathcal{A}$ is robust, then $\mathcal{B}$ is robust.
\end{theorem}
\begin{proof}
We use the straightforward approximation-preserving reduction of Garg et al.~\cite{DBLP:journals/jal/GargVY04}. Let $G = (V, E, w)$ be a $\gamma$-stable instance of Minimum Vertex Cover, with $V = \{u_1, ..., u_n\}$. We construct $G' = (V', E', w')$, where $G'$ contains the whole graph $G$, and moreover, for each vertex $u_i \in V$, we create a terminal vertex $s_i$ and we connect it to $u_i$ with an edge $(s_i, u_i) \in E'$. As implied, the set of terminals is $T = \{s_1, ..., s_n\}$. The weights of non-terminal vertices remain unchanged. This is clearly a polynomial-time reduction. We will now prove that each feasible vertex cover $X$ of $G$ corresponds to a feasible Mulitway Cut of $G'$ of the same cost, and vice versa. To see this, let $X$ be a feasible vertex cover of $G$, and let's assume that there is a path between two terminals $s_i$ and $s_j$ in $G'[V' \setminus X]$. By construction, this means that there is a path between $u_i$ and $u_j$ in $G'[V' \setminus X]$, which implies that there is at least one edge in this path that is not covered. Thus, we get a contradiction. Since the weight function is unchanged, we also conclude that $w(X) = w'(X)$. Let now $X'$ be a feasible Multiway Cut for $G'$, and let's assume that $X'$ is not a vertex cover in $G$. This means that there is an edge $(u_i, u_j) \in E$ such that $\{u_i, u_j\} \cap X' = \emptyset$. This means that the induced graph $G'[V' \setminus X']$ contains the path $s_i - u_i - u_j - s_j$, and so we get a contradiction, since we assumed that $X'$ is a feasible Node Multiway Cut. Again, the cost is clearly the same, and thus, we conclude that there is a one-to-one correspondence between vertex covers of $G$ and multiway cuts of $G'$.

Since the cost function is exactly the same, it is now easy to prove that a $\gamma$-stable instance $G$ of Vertex Cover implies that $G'$ is a $\gamma$-stable instance of Multiway Cut, and moreover, if $G'$ is not $\gamma$-stable, then $G$ cannot be $\gamma$-stable to begin with. Thus, we can run algorithm $\mathcal{A}$ on instance $G'$, and return its output as the output of algorithm $\mathcal{B}$. By the previous discussion, this is a $\gamma$-stable algorithm for Vertex Cover, and, if $\mathcal{A}$ is robust, then so is $\mathcal{B}$.
\end{proof}

The above result, combined with the results of~\cite{DBLP:conf/stoc/AngelidakisMM17} and the results of Section~\ref{sec:hardness}, implies Theorem~\ref{thm:lower-bounds-mc}.

\end{document}